\documentclass[10pt,journal,compsoc]{IEEEtran}
\IEEEoverridecommandlockouts

\ifCLASSOPTIONcompsoc
  \usepackage[nocompress]{cite}
\else
  \usepackage{cite}
\fi

\usepackage{amsmath,amssymb,amsfonts}
\usepackage{amstext}
\usepackage{amsthm}
\usepackage{amsbsy}
\usepackage{bbm}
\usepackage{graphicx}
\usepackage{graphics}
\usepackage{algorithm}
\usepackage[noend]{algorithmic}
\usepackage{cite}
\usepackage{multirow}
\usepackage{pdfsync}
\usepackage{epstopdf} 
\usepackage{color}
\usepackage{caption}
\usepackage{subcaption} 
\usepackage{enumitem}

\newtheorem{theorem}{{Theorem}}

\newtheorem{property}[theorem]{{Property}}
\newtheorem{proposition}[theorem]{{Proposition}}

\usepackage[labelformat=simple]{subcaption}

\title{
Network Coding Based Evolutionary Network Formation for Dynamic Wireless Networks 
}

\author{Minhae~Kwon,~\IEEEmembership{Student~Member,~IEEE}
        and~Hyunggon~Park,~\IEEEmembership{Senior~Member,~IEEE} 
\IEEEcompsocitemizethanks{\IEEEcompsocthanksitem M. Kwon and H. Park are with the Department of Electronic and Electrical Engineering, Ewha Womans University, Seoul, Republic of Korea (E-mail: minhae.kwon@ewhain.net, hyunggon.park@ewha.ac.kr), Corresponding author: Hyunggon Park. 
\IEEEcompsocthanksitem This work was supported by the National Research Foundation of Korea (NRF) grant funded by the Korea government (MSIT) (No. NRF-2017R1A2B4005041)}

}

\begin{document}



\IEEEtitleabstractindextext{
\begin{abstract} 
In this paper, we aim to find a robust network formation strategy that can adaptively
evolve the network topology against network dynamics in a distributed manner. 
We consider a network coding deployed wireless ad hoc network where source nodes are
connected to terminal nodes with the help of intermediate nodes. 
We show that mixing operations in network coding can induce packet anonymity that 
allows the inter-connections in a network to be decoupled. 
This enables each intermediate node to consider complex network
inter-connections as a node-environment interaction such that the Markov decision
process (MDP) can be employed at each intermediate node. The
optimal policy that can be obtained 
by solving the MDP 
provides each node with optimal amount of changes in
transmission range given network dynamics (e.g., the number of nodes in the range and  channel
condition). 
Hence, the network can be adaptively and optimally evolved by responding to the network dynamics. 
The proposed strategy is used to maximize  long-term utility, which  is achieved by considering both current network
conditions and future network dynamics. 
We define the utility of an action to include network
throughput gain and the cost of transmission power. 
We show that the resulting network of the proposed strategy eventually converges to
stationary networks, which maintain the states of the nodes. 
Moreover, we propose to determine initial transmission ranges 
and initial
network topology that can expedite the convergence of the proposed algorithm. 
Our simulation results confirm that the proposed strategy builds
a network which adaptively changes its topology in the presence of network dynamics.  
Moreover, the proposed strategy outperforms existing strategies  in
terms of system goodput and successful connectivity ratio.

%
%
%
%
%
%

\end{abstract}

}

\maketitle

\IEEEdisplaynontitleabstractindextext

\IEEEpeerreviewmaketitle

\IEEEraisesectionheading{\section{Introduction}}

The connected world which  began with representative services 
such as connected cars, networked unmanned aerial vehicles (UAVs), multi-robot systems, and the Internet of
things (IoT), results in networks with inherent dynamics.  
The network entities of such services generally have high mobility,  
which causes frequent
changes in member nodes associated with these networks.  
Moreover,  
the links between the network entities may be under unstable channel conditions 
with high link failure
rates. 
Hence, it is essential  to form robust networks against such dynamics by
adaptively reformulating inter-connections among network entities. 
However, solving this problem based on conventional centralized
solutions requires very high computational complexity such that it cannot be
practically considered. 
Rather, it can be solved by decentralized and spontaneous network formation
strategies that 
enable each network entity to make proactive and adaptive decisions on network
topology against the network dynamics. 
In  wireless ad hoc networks, for example, source nodes are connected
to terminal nodes via mobile intermedia nodes. 
Then, the network topology can be determined in
a distributed manner based on decisions of mobile intermediate nodes for
transmission ranges.
In order to overcome network dynamics, each mobile intermediate  
node may strategically change its transmission range. 
Such distributed strategies for network formation can be essential in circumstances
where only limited
infrastructures can be available, e.g., disaster networks or military networks. 


\begin{figure}[tb]
\centering
\begin{center}
\includegraphics[width = 8.5cm]{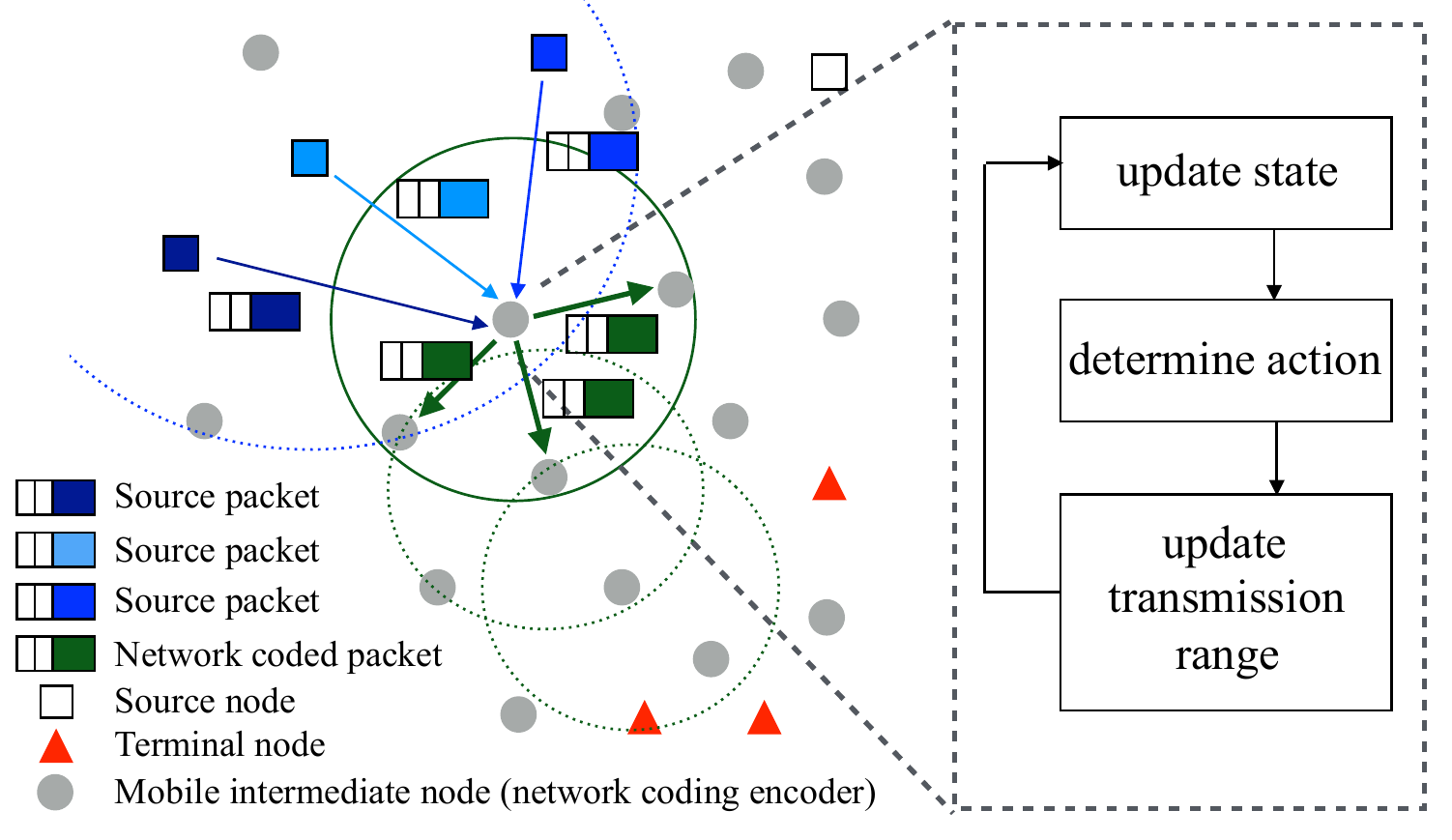}
\caption{
 An illustrative example of wireless ad hoc network where source nodes are connected
 to terminal nodes via mobile intermedia nodes  with network coding functionality. 
 }
 \vspace{-1cm}
\label{fig:example}
\end{center}
\end{figure}

However, it is not straightforward to design decentralized strategies that enable  
each network entity to make its \emph{own} and \emph{optimal} decisions, 
because the network entities are intimately coupled. 
Specifically, the network entities can be tightly inter-connected, so that 
the impact of small changes from a network entity may propagate over a 
large number of entities.
Thus, each network
entity should consider the corresponding responses 
associated with its decisions to make optimal decisions. This may require significantly high computational  complexity or
may not be feasible in practice. 
Therefore, it is essential for 
the design of
decentralized strategies to decouple the inter-connections among network entities. 

In this paper, we show that the inter-connections among network entities can be
decoupled  
by deploying network coding, which is referred to as  
\emph{network
decoupling}. 
Unlike the 
conventional store-and-forward approach,  
 network coding~\cite{Ahlswede2000} allows an intermediate entity to combine
multiple packets that it has received and to forward the \emph{mixed} packets.
For a network coding enabled wireless ad-hoc network (which is widely considered as a network model of a
connected world),  
a packet passes through many intermediate entities. Thus, it 
can be mixed with other packets multiple times. This
leads to \emph{packet
anonymity}, where all packets in the network eventually have identical information
including 
their terminal nodes. 
%
%
%
Packet anonymity allows an entity to consider the other entities as 
its \emph{environment}. 
Hence, complicated inter-connections among network entities can be decoupled, and 
only the connection directly associated with the entity is considered as  a one-hop connection. 
This leads to  \emph{network decoupling}  
so that the interactions between network entities can be
interpreted as a \emph{node-environment interaction} at each entity. 
An illustrative example of wireless ad hoc network with network coding is shown in Fig.~\ref{fig:example}. In this example, source nodes are connected
to terminal nodes via mobile intermedia nodes  with network coding functionality. 

%
%
%
%

Motivated by the node-environment interaction, we use 
an MDP
to find a decentralized strategy which is referred to as a policy 
for network formation. 
%
We consider wireless entities to be autonomous decision-making agents and the 
state of an agent is defined as the 
number of
effective nodes. 
Here, effective nodes are 
the entities that have \emph{successfully} received packets from the agent.
The probability density function of the states is modeled by the Poisson point
process (PPP), which is widely used to characterize the behavior of mobile nodes. 
The action of an agent is defined as the amount of increasing or decreasing
transmission range, which is the outcome of the policy 
for the
current state of the agent.  The  policy is
optimal if it enables the agent to  
maximize long-term utility. 


%
%
%

As a node
increases its transmission range, 
the number of hops required to 
reach the terminal node
decreases, without loss of generality,  leading to an improvement in network throughput. 
However,  extending the  
transmission range increases  transmission power consumption and causes more inter-node
interference. 
This is explicitly captured by the utility function, which represents  
both network throughput improvement 
and the additional  transmission power.   
Therefore, the optimal policy enables each entity to successively determine  
the optimal changes in transmission range at each state, such that 
the entities can strike a balance 
between network throughput gain and power consumption. 
Finally, the consequences of the distributed decisions from each entity 
eventually determine the 
network topology. 

Note that the proposed strategy allows the resulting network topology to 
evolutionarily adapt 
against network dynamics. This is because the state is defined by 
the effective nodes, which are  
directly dependent on link failure rates (i.e., channel
condition) and node mobility (i.e., node distribution). 
For example, a 
larger transmission range may be required in a channel with 
higher link failure rates for 
the target number of effective nodes.
Similarly, 
an agent can increase its transmission range  to sustain connectivity
in the case of  
sparse node density. 
The proposed strategy is also robust against frequent changes in member nodes of  the considered network, which is widely observed in mobile networks. 
This is because the behavior of existing nodes is not affected by individual network members, instead it is only affected by \emph{the number} of effective nodes included in its own transmission range. 

%
%

Unlike conventional optimal solutions that focus on maximizing 
immediate utility, 
the 
proposed optimal policy determined by the MDP 
can provide a long-term strategy, which determines
actions by explicitly considering 
future dynamics in the network. 
Specifically, the actions taken by the optimal policy can 
maximize the long-term utilities, which are expressed as 
the sum of discounted utilities over time. 
The discount factor can be determined by considering the consistency of 
network conditions. Therefore, the actions determined by the proposed policy 
can consider both  current and future network dynamics.


The proposed system consists of two phases: initialization and adaptation. 
In the initialization phase, the optimal policy for each intermediate node is
found and the state can be initialized. 
As will be shown in this paper, optimal actions can lead the network formation result to certain topologies, 
referred to as 
\emph{stationary networks}. Hence, we design the initial network to be close to the stationary network. 
In the adaptation phase, each node adaptively and optimally changes its 
transmission range based on 
the optimal policy for the current state 
induced by network dynamics.

 
The main contributions of this paper are summarized as follows. 
\begin{itemize}
\item We show that network coding can  
  lead to packet anonymity where both the information and terminal of all packets in network asymptotically become identical, 
  \item We show that the packet anonymity of network coding allows inter-connections among network nodes to be decoupled into node-environment interactions at each node, which is referred to as network decoupling, 
\item We formulate the problem of network topology formation in an MDP framework 
  and provide a decentralized solution to the network formation strategy, 
\item 
The proposed strategy improves network robustness by adaptively rebuilding its topology in the presence of network dynamics which includes unstable channel conditions with high link failure rates, and high mobility of network nodes that causes frequent changes in member nodes associated with the considered  network, 

%
%

\item 
The proposed strategy is a foresighted strategy that chooses the action  maximizes a long-term utility by considering future network dynamics, 
\item The proposed strategy can determine the optimal transmission range  
  that balances  network throughput improvement and 
  transmission power consumption, 
\item 
The resulting network of the proposed strategy converges to  stationary networks,
  and
\item We propose how to initialize a network such that the speed of 
  convergence to the stationary network can be improved. 
\end{itemize}

Note that the focus of this paper is neither on the code design for network coding which has been extensively studied in prior works~\cite{wang2008,
Kim2009jsac, Khreishah2010, Bourtsoulatze2014TCOM, Bourtsoulatze2014TMM, Hulya2009,
Douik2016, Xie2016}, nor on perfect delivery which needs $100\%$ reliability. 
Rather, our focus is on robust network formation based on network coding, which can
proactively reform network topology against network dynamics in a decentralized
manner.

                   

The rest of the paper is organized as follows. In Section~\ref{sec:related_works}, we
briefly review related works. The wireless model for mobile users
and detailed process of
data delivery based on network coding are discussed in Section~\ref{sec:system}. The
MDP-based framework and a distributed network formation 
strategy are proposed
in Section~\ref{sec:MDP} and Section~\ref{sec:distributed}, respectively. Simulation
results are presented in Section~\ref{sec:simulation}, and conclusions are drawn in
Section~\ref{sec:conclusion}.  

For the reader's convenience, we summarize notations frequently used in this
paper in Table~\ref{tab:notation}. 
\begin{table*}[tb]
\caption{Summary of Notations }
\label{tab:notation}
\renewcommand{\arraystretch}{1}
\begin{center}
{\begin{tabular}{  l | p{5.3cm} || l | p{5.3cm}  }
\hline
        {\textbf{Notation}}& \textbf{Description} & \textbf{Notation} & \textbf{Description}   \\
		\hline
\hline
         {$v_i$}& a node with index $i$& $\mathbf S$ & state space  \\
		\hline
         {$\mathbf H_t$}& an index set of source nodes for a terminal node $v_t$&$s$ & a state \\
		\hline
      {$\mathbf H$}&an index set of source nodes & $\mathbf A$ & action space\\
		\hline
       {$\mathbf T_h$}&an index set of terminals for a source node $v_h$ for $h \in \mathbf H$&$a$& an action \\
		\hline
        $\mathbf T$& total index set of all terminals in $\mathcal G_{\tau}$& $\bar a_{ \tau}$ & transmission range at time $\tau$\\
		\hline
    	 $\mathbf V$& an index set of intermediate nodes & $\rho$&discount factor\\
		\hline
 {$\lambda$}& node density in network&  $\omega$ & weight in \eqref{eqn:U} \\		
  \hline
 {$\Gamma_{i, \tau}(v_j, \mathbf T_h) $} & node value function & $R (s, s')$ &reward\\
  		\hline
  {$\delta_{i,\tau}(v_j)$}& the Euclidean distance from node $v_i$ to node $v_j$ at time $\tau$ & $\pi$& policy\\
		\hline
		 {$\bar \delta_{i,\tau}$}&the radius of intermediate node $v_i$'s transmission range at time $\tau$& $V_{\tau}(s)$ & state-value function\\
		\hline
       {$\Delta_{\tau}(v_j, v_t)$}& the number of hops between $v_j$ and $v_t$ at time $\tau$&$\boldsymbol \sigma$&limiting distribution\\
		\hline  	       
        {$\boldsymbol\Delta_{\tau}(v_j, \mathbf T_h)$}&a set of $\Delta_{\tau}(v_j, v_t)$ for all $t \in \mathbf T_h$ & $\sigma_s$ & limiting probability of state $s$  \\
		\hline
$\Phi$& network coding function&$\mathbf P$ &  state transition matrix\\
		\hline
 \end{tabular}}
\end{center}
\vspace{-0.5cm}
\end{table*}

%
%
%
%


\section{Related Works}
\label{sec:related_works}


Since network coding was  first introduced in \cite{Ahlswede2000}, it has shown
excellent ability to improve throughput, robustness and complexity. 
The beginning of network coding was for throughput gain in a multicast
scenario. In~\cite{Ahlswede2000}, it is shown that network coding can achieve
maximum throughput via the max-flow min-cut theorem, and it is further proved
that linear network coding can achieve the upper bound of capacity in~\cite{Li2003}.
Many works in 
network coding have been studied for random linear network coding
(RLNC)~\cite{Ho2003} as it is a simple, randomized encoding approach that is 
decentralized~\cite{RandomizedNC03, PracticalNC03}. 
As well as throughput gain, it has been shown that 
network coding also enhances robustness against packet loss  
in lossy wireless networks~\cite{Koetter2008, Zhang2008information,
Esmaeilzadeh2017}.

Another advantage of network coding is that there is a lower complexity requirement 
for network formation
compared to a conventional store-and-forward approach. 
In a conventional store-and-forward approach, it is difficult 
to find the optimal routing path
that can achieve the capacity upper bound. Even though an optimal routing
solution exists in some cases, such as the Steiner tree in multicast routing, finding the solution  is
still very complex within a centralized setting~\cite{Jain2003}. Network coding,
however, can transform complex network formation problems into 
 low-complexity distributed problems.  
For example, a distributed solution 
that satisfies optimality condition to minimum delay and minimum energy
consumption is proposed in \cite{Cui2004}. 
Decentralized algorithms for network formation that can minimize cost
per unit capacity are proposed in~\cite{Lun2005}.

Even though network coding can reduce complexity in general, it is known
that finding an optimal solution in network coding with multiple multicasts is an NP-hard
problem~\cite{Yan2006}. Hence, suboptimal but practical solutions are often 
studied \cite{traskov2006, mhkwon2017SPL}. A well-known
practical solution to 
network formation is proposed based on linear optimization in~\cite{traskov2006}.  
In~\cite{mhkwon2017SPL}, a distributed network formation solution is developed for network coding
deployed wireless networks that includes multi-source multicast flows. Using a 
game theoretical approach, each node in the network determines its transmission power and the
use of network coding operations.  

Network formation strategies for dynamic network conditions in conventional
routing schemes have been widely studied in the context of a self-organizing network. 
Protocols for self-organization of wireless sensor
networks where there exists a large number of static nodes with energy constraints are described in  \cite{sohrabi2000}.  
In~\cite{Gao2016}, an emergency communication system based on UAV-assisted
self-organizing network is considered. In this work, UAVs are used as a strong relay
node to form a relay network in the air, and the nodes on the ground formed a
self-organizing network automatically with the help of UAVs. 
In~\cite{xu2016}, a dynamic topology control that prolongs the lifetime of a wireless
sensor network is proposed based on a non-cooperative game.

However, there have been few studies in network coding deployed network formation
strategies that are robust in the presence of network dynamics.


\section{Network Coding Deployed Wireless Ad Hoc Networks}
\label{sec:system}

\subsection{Wireless Ad Hoc Network Model }

We consider a  wireless ad hoc network modeled by a directed graph $\mathcal G_\tau$, 
comprising a set $\mathcal V(\mathcal G_\tau)$ of nodes together with a set $\mathcal
E (\mathcal G_\tau)$ of directed links at time  $\tau$. 
Let $v_i \in \mathcal V(\mathcal G_\tau)$ be the $i$-th element in $\mathcal V(\mathcal G_\tau)$, and there are three types of nodes in the network: source, intermediate and terminal.
Let $\mathbf H$ be an index set of source nodes and its element is denoted as $h \in \mathbf H$. An index set of terminals for a source node $v_h$ is denoted as $\mathbf T_h$, and the data that $v_h$ generates at time $\tau$ are denoted as $x_{h,\tau}$. 
In this paper, we  consider the multi-source multicast flows that have multiple
source nodes, and each source node has an independent set of terminal nodes. 
Specifically, a source node $v_h$ for $h \in \mathbf H$ aims to deliver its data
$x_{h,\tau}$ to multiple terminal nodes $v_t, \forall t \in \mathbf T_h$ so that
$\sum_{h \in \mathbf H}| \mathbf T_h|$ flows are simultaneously considered, where $|
\cdot |$ denotes the size of a set. 
The total index set of all terminals in $\mathcal G_\tau$ is denoted as $\mathbf T = \bigcup_{h \in \mathbf H} \mathbf T_h$, and  
the number of source nodes and terminal nodes are denoted by $N_H$ and $N_T$, respectively. 
The source nodes can only transmit data but cannot 
receive data, and the terminal nodes can only receive data but cannot transmit data,
unlike the source nodes. 

In cases where the source nodes are not able to directly transmit data to the terminal
nodes,
intermediate nodes can relay the data by receiving and transmitting the data.
Let $v_i$ for $ i \in \mathbf V$ be an intermediate node where $\mathbf V$ denotes an
index set of intermediate nodes, and $N_V$ be the number of intermediate nodes. Then,
the total number of nodes in the network can be represented by $|\mathcal V(\mathcal
G_\tau)|=N_H + N_V + N_T$. 

We consider the intermediate nodes as wireless mobile devices 
that can move around in a bounded region   
with energy constraints. 
We use a stochastic geometry model 
to capture the distribution of intermediate nodes  to model the characteristics of mobility; the number of intermediate nodes
in a bounded region follows an independent homogeneous PPP with node density
$\lambda$, which is 
the expected number of Poisson points~\cite{Andrews2010,
Huang2014}. Moreover, each intermediate node can adjust its transmission power, which
determines the range of potential delivery of data from the node, which is referred to as
\emph{transmission range}. 
Let $\bar \delta_{i,\tau}$ 
and $\delta_{i, \tau}(v_j)$ be 
the radius of the transmission range
of $v_i$ and the Euclidean distance from node $v_i$ to node $v_j$
at time $\tau$, respectively. 
Then, $v_j$ is
located in the transmission range of $v_i$
if  $\delta_{i, \tau}(v_j) \le \bar \delta_{i,\tau}$ and $v_j$ can receive the data
from $v_i$. 
We assume that links between nodes may be disconnected with probability $\beta$,
which is referred to as the link failure rate.  
Hence, the probability that 
$v_j$ in the transmission range of $v_i$ 
can successfully receive the data
from $v_i$ is given by $1-\beta$.


The deployment of intermediate nodes naturally leads to a multi-hop ad hoc network, 
and thus, the network throughput highly depends on the selection of paths constructed
by the nodes. 
Therefore, node $v_i$ needs to choose a node $v_j$, which can relay the data to terminal
nodes $\mathbf T_h$ better than
other neighbor nodes in terms of node values.
The node value of $v_j$ from the perspective of $v_i$
is evaluated by the  
\emph{node
value function} 
$\Gamma_{i, \tau}(v_j, \mathbf T_h)$,
expressed as
\begin{equation}
\Gamma_{i, \tau}(v_j, \mathbf T_h) = f(\delta_{i, \tau}(v_j), \boldsymbol\Delta_{ \tau}(v_j, \mathbf T_h)) \label{eqn:node_value_function}
\end{equation}
where $\boldsymbol\Delta_{ \tau}(v_j, \mathbf T_h)= \{ \Delta_{ \tau}(v_j, v_t) |
\forall t \in \mathbf T_h \}$ and $\Delta_{ \tau}(v_j, v_t)$ denotes the number of
hops between $v_j$ and $v_t$ at time $\tau$. The function $f(\delta_{i, \tau}(v_j),
\boldsymbol\Delta_{ \tau}(v_j, \mathbf T_h)) : (\mathbb R, \mathbb W^{|\mathbf
T_h|\times 1}) \rightarrow \mathbb R $ 
 is a decreasing function of  $\delta_{i, \tau}(v_j)$ and $\boldsymbol\Delta_{ \tau}(v_j, \mathbf T_h)$, where $\mathbb R$ denotes the field of real numbers and 
 $\mathbb W = \{ 0, \mathbb
 Z^+\}$ denotes the whole numbers which includes zero and the positive integers
 $\mathbb Z^+$. 
By defining the node value as a decreasing function of the distance and the number of
hops from $v_j$ to
$v_t$, 
$\Gamma_{i, \tau}(v_j, \mathbf T_h)$ increases
as $v_j$ is located closer to $v_i$, and $v_j$ is connected to $v_t, \forall t \in
\mathbf T_h$ with a smaller number of hops. Therefore, 
$v_i$ consumes lower transmission power 
with a smaller transmission range and reduces delay for data delivery for appropriate selection of 
$v_j$. 

\subsection{Network Coding Based Encoding Process }
\label{subsec:NC}


A source node $v_h$ for $h \in \mathbf H$ generates a set of data $x_{h,\tau}=
\left\{x_{h,\tau}^{(1)}, \ldots, x_{h,\tau}^{(L)} \right\}$ at time $\tau$ and it
broadcasts
$x_{h,\tau}$ with the transmission power $\bar {a}_h=g(\bar \delta_h)$,
where the function $g: \mathbb R \rightarrow \mathbb R$ is determined 
based on a path loss
model of wireless channels.  
We assume that the radius of transmission range $\bar \delta_h$ of $v_h$
is stationary
(i.e., time independent), so that the subscription $\tau$ is omitted. If an
intermediate node $v_i$ is located in the transmission range of $v_h$ at time $\tau$,
$v_i$ receives $x_{h,\tau}$, and puts $x_{h,\tau}$ into its buffer $\mathcal L_{i}$,
i.e., $x_{h,\tau} \in \mathcal L_{i}$ wherein data are sorted by time stamp $\tau$
with the oldest time stamp at the head of the queue~\cite{chou2007}. 
Note that the packet has a limited life span (e.g., time to live (TTL) in an internet
packet) such that the packets with an expired time stamp can be discarded. 
%
For simplicity, we assume that the output capacity of a node is a single packet size such that a node  transmits a single packet per unit time~\cite{Katti2006, nad2004, topakkaya2011},
and a node can receive multiple individual packets by applying multipacket reception techniques~\cite{Cloud2012, mirrezaei2014}.



The intermediate node $v_i$ performs network coding operations by combining packets
with the same time stamp  
in $\mathcal L_{i}$ and generates encoded data $y_{i, \tau+1}$ at time $\tau+1$
expressed as  
\begin{align}
y_{i,\tau+1} 
&=\sum_{h=1}^{N_H}\bigoplus  \left( C_{hi, \tau+1} \otimes x_{h,\tau} \right) 
\label{eqn:NC_source}
\end{align}
where $C_{hi, \tau+1}$  denotes the global coding coefficient of $v_i$ for source
data
$x_{h, \tau}$.  The network coding operations are performed in the Galois field (GF)
and the operators $\oplus$ and $\otimes$ denote the addition and multiplication in
GF,
respectively. 
When the encoding process is performed in~\eqref{eqn:NC_source}, the source data with
the same time stamp are combined together, and a packet $p_{i,\tau+1}$ is constructed
as 
\begin{equation*}
p_{i,\tau+1} = [\tau, C_{1i, \tau+1}, \ldots, C_{N_Hi, \tau+1}, y_{i,\tau+1}]
\end{equation*}
which has the time stamp of the combined source data $\tau$, the global coding
coefficient $C_{hi, \tau+1}, \forall h \in \mathbf H$ as the header, and the encoded
data $y_{i,\tau+1}$ as a payload.
An index set of terminals for $p_{i,\tau+1}$ denoted by $\mathbf T_{p_{i,\tau+1}}$
can be expressed as  
\begin{equation}
\mathbf T_{p_{i,\tau+1}} = \cup_{h \in \{ h| C_{hi, \tau+1} \ne 0, h \in \mathbf H
\}} \mathbf T_h.
\label{eqn:T_p}
\end{equation}
This is because $p_{i,\tau+1}$ needs to be delivered to all terminals of combined
source data, 
i.e., for all $v_t \in \mathbf T_h$ and $h \in \{ h| C_{hi, \tau+1} \ne 0,
\forall h \in \mathbf H \}$.  


If the intermediate node $v_i$ receives the encoded packets $y_{h, \tau+ \alpha}$ at
time $\tau+\alpha$ , it recombines the received data and generates the encoded data
$y_{i, \tau+\alpha+1}$ at time $\tau+\alpha+1$, i.e.,  
\begin{align}
&y_{i,\tau+\alpha+1} \notag\\
&= \sum_{y_{j,\tau+\alpha} \in \mathcal L_{i}} \bigoplus \left( c_{ji} \otimes y_{j,\tau+\alpha} \right) \label{eqn:NC1}
\\
&= \sum_{y_{j,\tau+\alpha} \in \mathcal L_{i}} \bigoplus \left( c_{ji} \otimes \left(\sum_{h=1}^{N_H}\bigoplus  \left( C_{hj, \tau+\alpha} \otimes x_{h,\tau} \right) \right) \right)\notag\\
&= \sum_{h=1}^{N_H}\bigoplus \left( \sum_{y_{j,\tau+\alpha} \in \mathcal L_{i}} \bigoplus (c_{ji} \otimes C_{hj, \tau+\alpha}) \right) \otimes x_{h,\tau}\notag\\
&= \sum_{h=1}^{N_H}\bigoplus  \left( C_{hi, \tau+\alpha+1} \otimes x_{h,\tau} \right) \notag
\end{align}
where $\alpha > 0$,  and 
$c_{ji}$\footnote{The time stamp for $c_{ji}$ is omitted 
because the local coding coefficient is used only for one time slot.} denotes the local
coding coefficient for data from $v_j$ to $v_i$. In this paper, network coding is
implemented based on RLNC~\cite{HO2006} so that $c_{ji}$ is
uniformly and randomly chosen from GF with a size of $2^M$ (GF($2^M$)), i.e., $c_{ij} \in$ GF($2^M$). 
However, the proposed strategy is not limited to RLNC, and deterministic code
designs~\cite{wang2008,
Kim2009jsac, Khreishah2010, Bourtsoulatze2014TCOM, Bourtsoulatze2014TMM, Hulya2009,
Douik2016, Xie2016} can be considered as well.

For a large-scale multi-hop wireless network, 
the process of recombining incoming packets in~\eqref{eqn:NC1} can be performed
significantly many times,  
which eventually allows the recombined packet  
to include all 
source data. Therefore, 
the terminal  set of all packets in the network asymptotically converges to $\mathbf T$
by making all packets  identical. This is defined as \emph{packet anonymity} of network
coding, which is expressed in Proposition~\ref{def:pkt_ano}.

\begin{proposition}[Packet Anonymity]
  \label{def:pkt_ano}
Network coding can asymptotically make 
both the 
information and terminal of each packet identical. 
\end{proposition}

We next consider the impact of the packet anonymity 
on the node value. 
Let $\Phi$ be the network coding function in \eqref{eqn:NC1}. Then, the node value
function $\Gamma_{i, \tau}(v_j, \mathbf T_h)$ in~\eqref{eqn:node_value_function} that is 
transformed by the network coding function $\Phi$ can
be expressed as 
\begin{align}
\Phi(\Gamma_{i, \tau}(v_j, \mathbf T_h)) &= \Phi \left(f(\delta_{i, \tau}(v_j), \boldsymbol\Delta_{ \tau}(v_j, \mathbf T_h))\right) \label{eqn:phi_V_1} \\
&= f(\delta_{i, \tau}(v_j), \boldsymbol\Delta_{ \tau}(v_j, \mathbf T)) \label{eqn:phi_V_2}\\
&= \Gamma_{i, \tau}(v_j, \mathbf T). \label{eqn:NC_NVF}
\end{align}
The equality between \eqref{eqn:phi_V_1} and \eqref{eqn:phi_V_2} is based on packet
anonymity. $\mathbf T$
in \eqref{eqn:phi_V_2} is constant, so that the node value is only a function of $v_j$, 
which concludes \eqref{eqn:NC_NVF}. 
Therefore, we conclude that if the network coding function $\Phi$ is
employed in intermediate nodes, 
multi-hop
connections to terminals (i.e.,  $\mathbf T_h$) need not be considered for node $v_i$. 
Rather,  only 
the links directly associated with it should be considered 
as in a one-hop
connection (i.e., $v_j$), leading to   
\emph{network decoupling} described in Proposition~\ref{def:nc_decouple}. 
\begin{proposition}[Network Decoupling] 
  \label{def:nc_decouple}
Network coding can decouple one-hop connections from the 
overall network formation. 
\end{proposition}
Network decoupling can lead to the design of decentralized solutions by only considering 
node-environment interactions at each node, which are captured by 
the MDP framework 
introduced in Section~\ref{sec:MDP}. 
The characteristics of network coding function $\Phi$ are shown in  
Fig.~\ref{fig:NC_effect}. 
\begin{figure}[t]
\centering
\begin{center}
\includegraphics[width = 8.5cm]{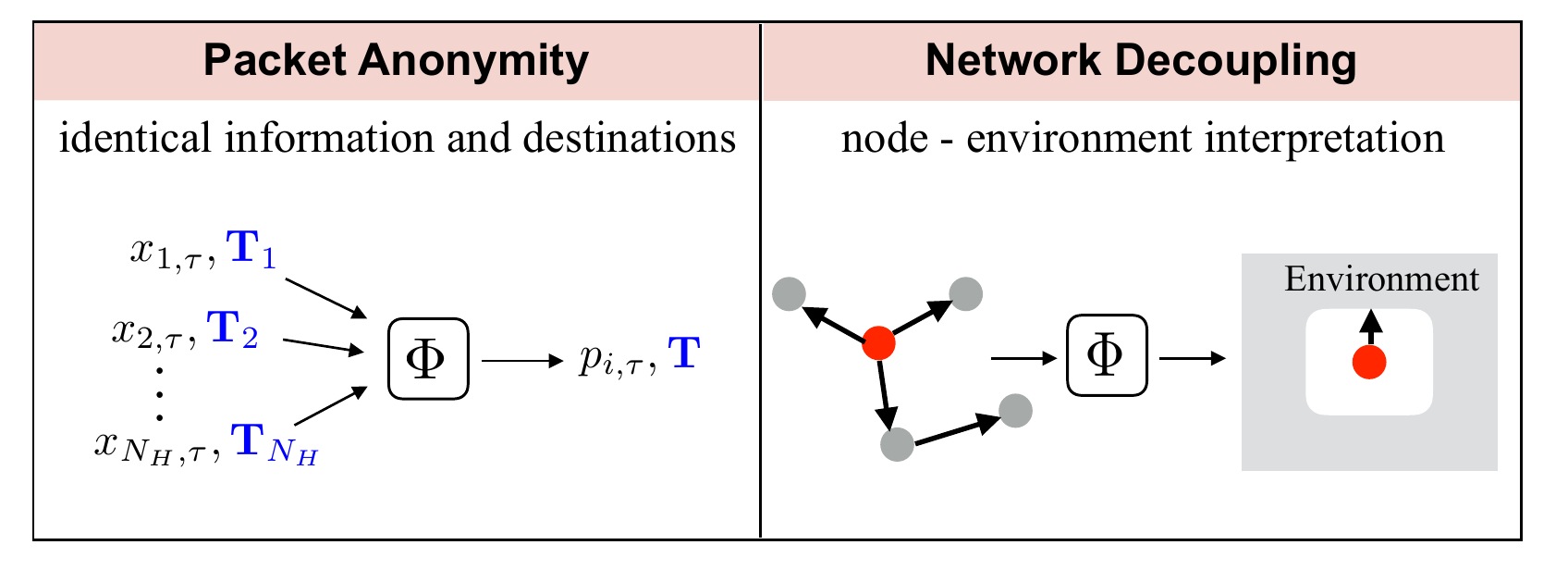}
\caption{ Two characteristics of the network coding function $\Phi$. }
\label{fig:NC_effect}
\end{center}
\vspace{-1cm}
\end{figure}

\subsection{Source Reconstruction at Terminal Nodes }

We next discuss the decoding process. 
Let 
$\mathbf H_t = \{ h | t \in \mathbf T_h, \forall h \in \mathbf H \}$ 
be 
an index set of source nodes for a terminal node $v_t$ 
and $ \tilde p_1, \cdots, \tilde p_K \in \mathcal L_t$ be the packets with the same
time stamp of source data that $v_t$ received. 
Then, we can construct 
a vector of network coded data $\tilde {\mathbf y} =
[ \tilde y_1, \cdots, \tilde y_K  ]^T$ and 
the global coding coefficient
matrix $\tilde {\mathbf C}$ is expressed as 
\begin{align}
\tilde {\mathbf C} 
&= 
\begin{bmatrix}
    \tilde {C}_{11} &  \cdots & \tilde {C}_{N_H 1} \\
    & \vdots & \\
    \tilde {C}_{1K} &  \cdots & \tilde {C}_{N_H K}
        \end{bmatrix}
        = [\tilde {\mathbf c}_1, \cdots, \tilde{\mathbf c}_{N_H} ] 
\end{align}
where $\tilde{\mathbf c}_h = [\tilde {C}_{h1}, \cdots, \tilde {C}_{hK} ]^T$ for $1 \le h \le N_H$.   

Node $v_t$ is then able to perfectly reconstruct its source
data, if $\tilde {\mathbf C}$
satisfies the following two conditions:   
1) $\tilde{\mathbf c}_h \ne \mathbf 0_{K}$ for all $h \in \mathbf H_t$, 
where $\mathbf 0_{K}$ denotes all zero vector
with length $K$,
and 2) $\tilde
{\mathbf C}'$ is full-rank, where $\tilde {\mathbf C}'$ is the matrix 
where all $\tilde
{\mathbf c}_h = \mathbf 0_K$ for $1 \le h \le N_H$ are removed from $\tilde {\mathbf
C}$.
Condition 1) ensures that the received packets include all data that
should be
reconstructed. 
This is a widely accepted condition under the wireless network settings
because of the broadcasting nature of wireless communications~\cite{Katti2005, Liu2007infocom}. 
The condition 2) guarantees that the decoding process can uniquely reconstruct 
data $\hat x_h$ for all $h \in \mathbf H_t$. 
Because it is shown that RLNC makes the global coding coefficient matrix
be full-rank with high
probability~\cite{PracticalNC03, RandomizedNC03}\footnote{
It is shown in 
\cite{RandomizedNC03} 
that if RLNC is employed, the probability that
the global coding coefficient matrix is full-rank is at least
${(1-{|\mathcal D|}/{2^M})}^{|\mathcal E (\mathcal G)|}$. In general settings,
a GF size of $2^M$ is significantly larger than the number of terminals $|\mathcal D|$ in the 
network. Hence, it is widely accepted that the global coding coefficient matrix
is full-rank with high probability if RLNC is used.}, the condition 2) can be satisfied. 
The 
decoding process can then be implemented based on well-known approaches such as
Gaussian elimination in
a GF \cite{bathe1976}. 

While the conditions for perfect reconstruction can  generally be satisfied with high probability,  
some special applications (e.g., delay-sensitive applications, error-prone networks
with a high packet loss rate, etc.)
may cause a perfect reconstruction to fail. That is,  random mixing in the inter-session network
coding may lead to an increased decoding delay if only a subset of the coded
sources of interest  arrives at the terminal node. 
In this case, alternative decoding algorithms~\cite{ mhKWON2014WCNC,
Yan2013,mhKWON2016elsevier,KWON2016} can be deployed.

In the rest of this paper,  
we propose a distributed strategy for robust network formation.


\section{ MDP-Based Network Formation}
\label{sec:MDP}

%
%
%

In this section, we propose an MDP-based framework for network formation, where
intermediate nodes $v_i, \forall i \in \mathbf V$ of the network are considered as
autonomous decision making agents to find the optimal strategy.   
An illustrative overview for the proposed framework is shown in Fig.~\ref{fig:overview_all}.  
\begin{figure}[tb]
\centering
\begin{center}
\includegraphics[width = 8.5cm]{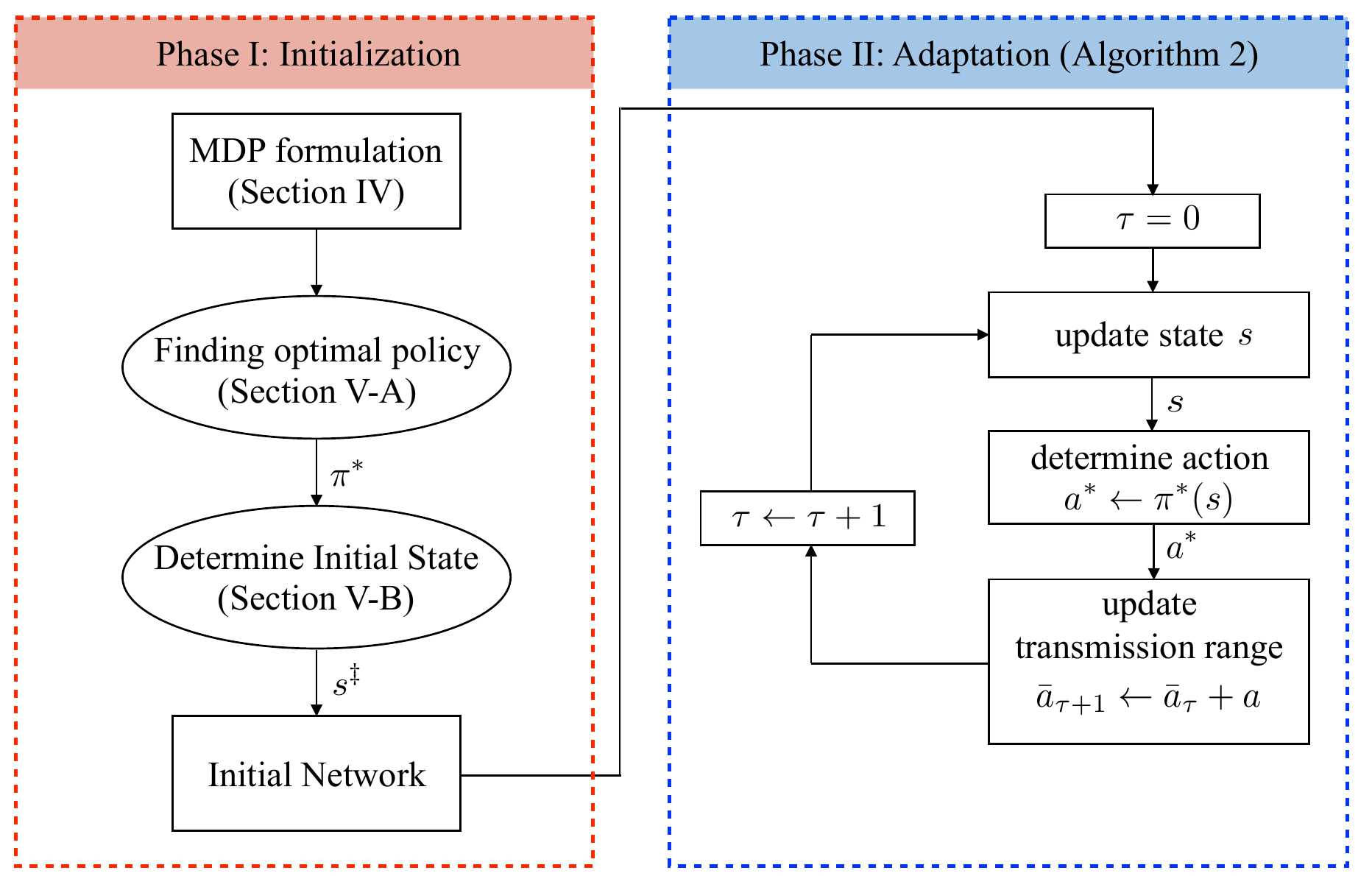}
\caption{ An illustrative overview of  the proposed system. }
\label{fig:overview_all}
\end{center}
\vspace{-0.8cm}
\end{figure}

For an agent $v_i$, an MDP is a tuple $\left< \mathbf S, \mathbf A, \right.$ $\left.
P(s'|s,a), U(s,a,s'), \rho  \right>$, where $\mathbf S$ is the state space, $\mathbf
A$ is the action space, and $P(s'|s,a): \mathbf S \times \mathbf A \times \mathbf S
\rightarrow [0,1]$ is the state transition probability that action $a \in \mathbf A$
in state $s \in \mathbf S$ leads to the next state $s' \in \mathbf S_i$, which is a
real number between $0$ and $1$. $U(s,a,s'): \mathbf S \times \mathbf A \times
\mathbf S \rightarrow \mathbb R $ is a utility obtained after transition to state
$s'$ from state $s$ with action $a$, and $\rho \in [0,1]$ is the discount factor. 
The details are explained as follows.

\subsubsection{State Space $\mathbf S$}
A state  $s \in \mathbf S$ represents the expected number of effective nodes in the
transmission range of the agent. 
Since a node can always be the effective node, at most $(N_V+N_T)$ nodes can be located  
in the transmission range, so that $1 \le s \le \lceil (N_V+N_T)/(1-\beta)\rceil$ in the channel
with a $\beta$ link failure rate.  

Note that the definition of the state allows the network to be robust against network
dynamics 
since a node can adaptively change its transmission ranges by considering node mobility and 
channel condition. 
If the network is static with node density $\lambda$, simply determining $\bar \delta_{i, \tau}$ can provide a solution for network topology, which directly determines the number of successfully received nodes, i.e.,  $s$, based on node density $\lambda$. 
If the network is dynamic, however, $s$ at time $\tau$ cannot be directly determined by 
$\bar \delta_{i, \tau}$ since $\lambda$ may not be a true value in the transmission range
of the agent because of node mobility and the link failure of the channel. Hence,  we 
design topology based on $s$ rather than $\bar \delta_{i, \tau}$, which allows a node to
adaptively change $\bar \delta_{i, \tau}$ against network dynamics,  leading to
a robust network.

\subsubsection{Action Space $\mathbf A$}
An action $a \in \mathbf A$ represents the increases in the 
transmission range as 
compared to a previous transmission range. Hence, the action at time $\tau$ becomes
$a = \bar a_{\tau} - \bar a_{ \tau-1}$, where $\bar a_{ \tau}$ and $\bar a_{ \tau-1}$
denote the transmission ranges at time $\tau$ and $\tau-1$, respectively. If
$a > 0$, the agent increases the 
transmission range 
(i.e., $ \bar a_{ \tau} > \bar a_{ \tau-1}$).
Similarly, if $a < 0$, the agent decreases the 
transmission range (i.e.,$ \bar a_{
\tau} < \bar a_{ \tau-1}$). The agent may keep the same transmission range by taking
action $a=0$.


\subsubsection{State Transition Probability $P(s'|s,a)$}
A state transition probability represents the probability that a node  
in state $s$
moves to state $s'$ if action $a$ is taken. 
Thus, $P(s'|s,a)$ means the probability that $s'$ effective nodes will be  
included in the
transmission range of the node in the next time stamp 
by taking action $a$ from current 
$s$ effective nodes.
Since the number of intermediate nodes in
a bounded region follows an independent homogeneous PPP
with node density $\lambda$, the state transition probability can be described as
in
Theorem~\ref{thm:STP}.
\begin{theorem}
\label{thm:STP}
The state transition probability $P(s'|s,a)$ is given by  
\begin{align}
P(s'|s, a) =
\begin{cases}
\frac{(\lambda a)^{\xi'-\xi} \cdot e^{-\lambda a}}{(\xi'-\xi)!}  & a > 0 \\
1 &  a = 0 \\
{\xi\choose{\xi'} }{(1- \frac{|a|}{\bar a_{ \tau}})}^{\xi'}{( \frac{|a|}{\bar a_{ \tau}})}^{\xi- \xi'}& a <0. 
\end{cases}
\label{eqn:stp_proof_re}
\end{align} 
where $\xi \triangleq \lceil \frac{s }{1-\beta} \rceil$ and $\xi' \triangleq \lceil \frac{s' }{1-\beta} \rceil$.

\end{theorem}

\begin{proof}
Based on 
the Kolmogorov definition of conditional probability,   
the state transition probability given in \eqref{eqn:stp_proof_re} can be expressed
as 
\begin{align}
P(s'|s, a) \cdot P(s) = P(s| s', -a) \cdot P(s').
\label{eqn:stp_aim}
\end{align}

%

Let $\xi $ be the number of nodes at time $\tau$ included in $\bar a_{\tau}$, and 
thus, the corresponding probability is given by 
\begin{align}
\Pr\{\xi ; \bar a_{ \tau}\} = \frac{(\lambda \bar a_{ \tau})^{\xi} \cdot e^{-\lambda \bar a_{ \tau}}}{\xi!}.  
\end{align}
By considering the link failure rate $\beta$ of the channel, 
the expected number of effective
nodes becomes $ s = (1-\beta) \cdot \xi$. Hence, $P(s) = \Pr\{\xi ;
\bar a_{ \tau}\}$ with integer value  
$\xi \triangleq \lceil \frac{s }{1-\beta} \rceil$. 

We assume that the 
state transition interval is short enough
under mild regularity
conditions~\cite{anderson2011}. 
Hence,  $a = 0$ (which does not change the transmission range)  does not lead to a state
transition, i.e., $s'=s$. 
Therefore, 
\begin{align*}
P(s'|s, a) \cdot P(s) &=  1 \cdot P(s) \\
&= 1 \cdot P(s') = P(s| s', -a) \cdot P(s'). 
\end{align*}

If $a > 0$ (which enlarges the transmission range), more nodes can be included in the 
transmission region. Hence, 
$s' \ge s$.
In this case,  $P(s'| s, a) \cdot P(s')$ in \eqref{eqn:stp_aim} can be derived as
\begin{align}
&P(s'|s,a) \cdot P(s) \notag\\
&= P(s'|s,a) \cdot \Pr\{\xi; \bar a_\tau\}\notag \\
&= \frac{(\lambda a)^{\xi'-\xi} \cdot e^{-\lambda a}}{(\xi'-\xi)!}  \cdot \frac{(\lambda \bar a_{\tau})^{\xi} \cdot e^{-\lambda \bar a_{\tau}}}{\xi!}\notag \\
&= \frac{(\lambda (\bar a_{\tau+1} - \bar a_{\tau}))^{\xi'-\xi} \cdot e^{-\lambda
(\bar a_{\tau+1} - \bar a_{\tau})}}{(\xi'-\xi)!}  \cdot \frac{(\lambda \bar
a_{\tau})^{\xi} \cdot e^{-\lambda \bar a_{\tau}}}{\xi!} \notag \\
&= \frac{(\lambda (\bar a_{\tau+1} - \bar a_{\tau}))^{\xi'-\xi} \cdot e^{-\lambda (\bar a_{\tau+1} )}(\lambda \bar a_{\tau})^{\xi} }
{(\xi'-\xi)! \xi!}
\cdot \frac{(\lambda \bar a_{\tau+1})^{\xi'}}{(\lambda \bar a_{\tau+1})^{\xi'}}
\cdot \frac{\xi'!}{\xi'!} \notag \\
&= \frac{(\lambda \bar a_{\tau+1})^{\xi'} e^{-\lambda (\bar a_{\tau+1} )} }{\xi'!} \cdot
\frac{\xi'!}{(\xi'-\xi)! \xi!} \cdot
(1-\frac{\bar a_{\tau}}{\bar a_{\tau+1}})^{\xi'-\xi} \cdot
(\frac{\bar a_{\tau}}{\bar a_{\tau+1}})^{\xi} \notag\\
&= \Pr\{\xi';\bar a_{\tau+1}\} \cdot {{\xi'}\choose{\xi}}\cdot (\frac{a}{\bar a_{\tau+1}})^{\xi'-\xi} \cdot (1 - \frac{a}{\bar a_{\tau+1}})^{\xi}\notag\\
&= P(s|s',-a) \cdot P(s'). \notag
\end{align}

Similarly,
$P(s'|s, a) \cdot P(s) = P(s| s', -a) \cdot P(s')$ 
in \eqref{eqn:stp_aim} is obtained for $a < 0$. 

Therefore, we conclude that the state transition probability can be expressed 
as \eqref{eqn:stp_proof_re}. 
\end{proof}
Theorem~\ref{thm:STP} implies that 
the state transition probability is the probability that
$(\xi' - \xi)$ nodes are included in the transmission range $a$ for 
$a >0$. 
If $a <0$, the probability that $\xi'$ nodes are in $\bar a_{\tau+1}$ 
given $\xi$ nodes in $\bar a_{\tau}$ is 
the probability that $(\xi- \xi')$ nodes are included in $|a|$.

\subsubsection{Utility Function $U(s,a,s')$}


We define the utility function of node $v_i$ as a quasi-linear function that 
consists of a reward and a cost, i.e., 
\begin{align}
U(s,a,s') = u+\omega \cdot R(s,s') - (1-\omega) \cdot  a
\label{eqn:U}
\end{align}  
where 
$R(s,s')$ is the reward function that represents immediate throughput improvement
given the state transition from 
$s$ to $s'$ at the cost of taking action $a$, which 
increases the transmission range. 
The cost intrinsically includes transmission power consumption at the node as well as the penalty for causing wireless inter-node interference. 
The weight $\omega$$(0\le \omega \le 1)$ can be used to balance 
the reward and the cost. For example, 
if $\omega = 1$, the cost associated with taking action $a$ can be 
ignored, but only the throughput improvement is  considered. 
Since a utility is generally non-negative, a constant $u$ is introduced in
\eqref{eqn:U}, and it can be
set such that $U(s,a,s')\ge 0$. 
The reward function $R(s,s')$ is defined as
\begin{align*}
R(s,s') = \gamma(s')-\gamma(s) 
\end{align*}
where $\gamma(s)$ denotes network throughput when the node is in state $s$, which is a concave increasing function. 


\subsubsection{Discount Factor $\rho$}
The discount factor $\rho \in [0,1]$ represents the 
degree of utility reduction over time, so that it determines the cumulative long-term
utility. The discount factor can be determined 
based on the consistency of the 
network condition (e.g.,~\cite{park2009spl, park2009tmm}).
For example, if the network condition is static, a large value of $\rho$ can be used by
imposing 
a high weight on the predicted future utilities whereas a lower value of $\rho$ needs to be
used in more dynamically changing network
conditions. 

Next, we show that the proposed framework satisfies the Markov property. 


\begin{theorem}
The tuple $\left< \mathbf S, \mathbf A, P(s'|s,a), U(s,a,s'), \rho  \right>$
satisfies the Markov property. 
\label{th:markov}
\end{theorem}

\begin{proof}
Let  $\left<s_{1},a_{1} \right>,  \left<s_{2},a_{2} \right>, \ldots,
\left<s_{\tau},a_{\tau} \right> $ be the sequence of events, where
$\left<s_{\tau},a_{\tau} \right>$ is an event which includes an action at
time $\tau$ (i.e., $a_{\tau}$) and a corresponding resulting state 
(i.e., $s_{\tau}$)\footnote{In this proof, we add time stamps $\tau$
on the notation of states and actions, e.g.,
$s_{ \tau}$ and $a_{\tau}$,  to clearly specify the time that an action is taken.}.  
The initial transmission range and corresponding state are denoted by $\bar a_0$ and $s_1$, respectively. 
To show that the tuple $\left< \mathbf S, \mathbf A, P(s'|s,a), U(s,a,s'),
\rho  \right>$ satisfies the Markov property, 
our aim is to prove 
%
\begin{align*}
P \left( s_{\tau+1}| \left<s_{\tau},a_{\tau} \right>, \ldots,  \left<s_{1},a_{1} \right> \right) 
= P(s_{\tau+1}|s_{\tau},a_{\tau}).
\end{align*}

The state transition probability that 
action $a_{1}$  
leads a node in state $s_{1}$ to a new state
$s_{2}$ can be expressed as 
\begin{equation*}
P \left( s_{2}|  s_1, a_{1} \right)
= \Pr\{s_1+(s_2-s_1) ; \bar a_0 +a_1 | s_1; \bar a_0\}
\end{equation*}
which implies that $(s_2-s_1)$ nodes are additionally included 
in the transmission range expanded by $a_1$.
Similarly, 
$P \left( s_{3}|  \left<s_{2},a_{2} \right>,  \left<s_{1},a_{1} \right> \right)$ can be expressed as

\begin{align*}
P ( s_{3}| & \left<s_{2},a_{2} \right>,  \left<s_{1},a_{1} \right> )\\
=  \Pr\{& s_1+(s_2-s_1)+(s_3-s_2) ; \bar a_0 +a_1 +a_2 \\ 
&|  s_1+(s_2-s_1); \bar a_0 +a_1  \}\\
=  \Pr\{&s_1+\sum_{t=1}^2(s_{t+1}-s_{t}) ;\bar a_0 +\sum_{t=1}^2 a_{t}\\
&| s_1+\sum_{t=1}^1(s_{t+1}-s_{t}) ; \bar a_0 +\sum_{t=1}^1 a_{t}\}\\
=\Pr\{&s_3-s_2; a_2\}= P(s_3| s_2, a_2 ).
\end{align*}
By induction, $P \left( s_{\tau+1}| \left<s_{\tau},a_{\tau} \right>,
\ldots,  \left<s_{1},a_{1} \right> \right)$ can be expressed as
\begin{align}
P ( s_{\tau+1}&| \left<s_{\tau},a_{\tau} \right>, \ldots,  \left<s_{1},a_{1} \right>  ) \notag \\
=  \Pr\{&s_1+\sum_{t=1}^\tau(s_{t+1}-s_{t}) ; \lambda(\bar a_0 +\sum_{t=1}^\tau a_{t})\notag\\
&| s_1+\sum_{t=1}^{\tau-1}(s_{t+1}-s_{t}) ; \lambda(\bar a_0 +\sum_{t=1}^{\tau-1} a_{t})\} \label{eqn:markov1}\\
= \Pr \{&s_{\tau+1}-s_{\tau}; \lambda \cdot a_{\tau} \}\label{eqn:markov2}\\
=P(s&_{\tau+1}|s_{\tau},a_{\tau})\notag. 
\end{align}
The equality between \eqref{eqn:markov1} and \eqref{eqn:markov2} can be derived 
by 
\begin{align*}
&s_{t+1}-s_{t} = \sum_{t=1}^{\tau}(s_{t+1}-s_{t}) -  \sum_{t=1}^{\tau-1}(s_{t+1}-s_{t})
\end{align*} 
and 
\begin{align*}
&\lambda\cdot a_{\tau} = \lambda (\bar a_0 + \sum_{t=1}^{\tau}a_{t} ) - \lambda (\bar a_0 + \sum_{t=1}^{\tau-1} a_{t}  )
\end{align*} 
which completes the proof. 
\end{proof}

Theorem~\ref{th:markov} shows that the proposed framework 
in this section can be modeled by the MDP. 
Fig.~\ref{fig:state} shows an illustration of the proposed framework. 
\begin{figure}[t]
\centering
\begin{center}
\includegraphics[width = 8.5cm]{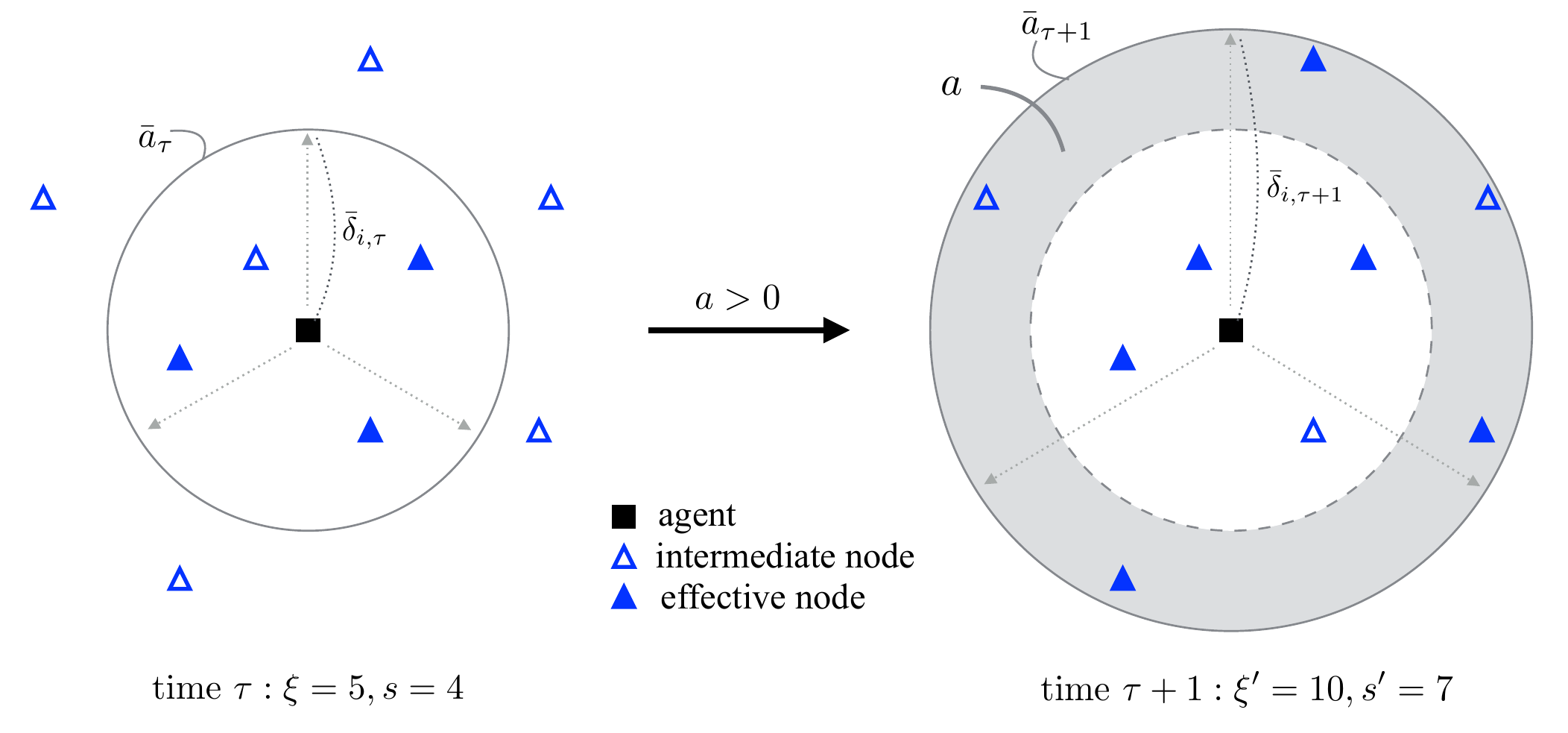}
\caption{An illustrative explanation of the proposed framework. 
}
\label{fig:state}
\end{center}
\vspace{-1cm}
\end{figure}

In the next section, we show how the strategy enables each node to 
make its own optimal decisions. 

\section{Distributed Network Formation Strategy}
\label{sec:distributed}


\subsection{MDP-Based Optimal Strategy for Network Formation}
\label{subsec:value_it}
 

 
The solution to an MDP is the optimal policy that maps the optimal actions performed in a particular state. 
Specifically, the policy is a function $\pi: \mathbf S \rightarrow \mathbf A$ which
returns an action for a state, i.e,  $\pi(s)= a$. 
The policy $\pi$ is optimal if it can maximize the 
\emph{state-value function} $V_{\tau}(s)$. 

The state-value function $V_{\tau}(s)$ represents a cumulative utility at time $\tau$, 
starting from state $s$, expressed as
\begin{align}
V_{\tau+1}(s) 
&= U(s,a,s') + \rho U(s',a, s'' ) + \rho^2 U(s'' ,a, s''') + \cdots  \notag\\
&=  U(s,a,s') + \rho V_{\tau}(s')   \label{eqn:return}
\end{align} 
where state $s$ sequentially moves into $s'$, $s''$, and $s'''$.  
In \eqref{eqn:return}, 
$V_{\tau}(s)$ 
includes the \emph{immediate utility} $U(s,a,s')$ and the discounted state-value of
successive  
states $\rho V_{\tau}(s')$. The expected value for the state-value function is thus
expressed as
\begin{align}
\mathbb E(V_{\tau+1}(s) )
=  \sum_{s'\in \mathbf S} P(s'|s,a)  \left(  U(s,a,s') + \rho   V_{\tau}(s') \right).
\label{eqn:E_V}
\end{align}
The \emph{optimal state-value function} $V^*(s)$ is the maximum state-value 
function over
all policies, i.e, 
\begin{equation}
V^*(s) = \max_{\pi} \mathbb E(V^{\pi}(s) )
\label{eqn:opt_value}
\end{equation} 
where $V^{\pi}(s)$ is the state-value achieved by the actions determined 
by policy $\pi$ at every state. 
Finally, the optimal policy $\pi^*$ is the policy that leads to $V^*(s)$, and it is defined as    
\begin{align}
\pi^*  &= \arg\max_{\pi}\mathbb E( V^{\pi}(s)), \forall s \in \mathbf S \label{eqn:value}
\end{align} 
This is also known as the Bellman optimality equation~\cite{bellman2015}. 
Given the optimal policy $\pi^*$, 
optimal action $a^*$ for each state $s$ can be determined such that 
\begin{align}
a^* &= \pi^*(s)\notag\\
&= \arg\max_{a \in \mathbf A} \sum_{s' \in \mathbf S}P(s'|s,a)\left( U(s,a,s') + \rho    V^*(s') \right). \notag
\end{align}

In practice, a near-optimal policy is widely used as it requires  
lower computational complexity. 
$\pi_{\epsilon}^*$ is an $\epsilon$-optimal policy if 
\begin{equation*}
  || V^{\pi_{\epsilon}^*}(s) - V^*(s) ||_{\infty} \le \epsilon 
\end{equation*}
which means that the error between $V^{\pi_\epsilon^*}(s)$, the state-value derived
by $\pi_\epsilon^*$, 
and $V^*(s)$
is bounded by the optimality level $\epsilon$.
The $\epsilon$-optimal policy can be found using 
Algorithm~\ref{alg:value_it}. 

\floatname{algorithm}{Algorithm}
\algsetup{indent= 1em}
\begin{algorithm}[t]
        \caption{Algorithm for $\epsilon$-Optimal Policy }
        \label{alg:value_it}
\begin{algorithmic}[1]
\smallskip
\REQUIRE{ 
state space $\mathbf S$, 
action space $\mathbf A$, 
utility function $U(s,a,s')$, 
weight $\omega$, 
discount factor $\rho$, 
state transition probability $P(s'|s,a)$, 
optimality level $\epsilon$} 
\STATE \textbf{Initialize:} $V_0(s) \leftarrow 0, \forall s \in \mathbf S$, $\tau \leftarrow 0$
\WHILE { $V_{\tau}(s)-V_{\tau-1}(s)> \frac{1-\rho}{2\rho} \epsilon$ for any $s\in \mathbf S$}
\FOR {$\forall s \in \mathbf S$  }
\STATE {update $V_{\tau+1}(s) \leftarrow \max_{a \in \mathbf A}$ $\left( \sum_{s'\in \mathbf S} P(s'|s,a) \right.$\\
\quad \quad \quad \quad \quad \quad \quad  \quad $ \left. \times (  U(s,a,s') + \rho   V_{\tau}(s') )\right)$}
	\ENDFOR
	\STATE{ $\tau \leftarrow \tau+1$ }
\ENDWHILE
\STATE {choose $\epsilon$-optimal policy $\pi^{\epsilon^*}(s) \leftarrow  \arg\max_{a} V_{\tau}(s), \forall s\in \mathbf S$}
\RETURN{  $\epsilon$-optimal policies $ \pi^{\epsilon^*}$} 
\end{algorithmic}
\end{algorithm}





%


Using stopping criteria, we show that  
Algorithm~\ref{alg:value_it} converges to $\epsilon$-optimal policies
in Theorem~\ref{thm:stop_cond}.

\begin{theorem}
  \label{thm:stop_cond}
The $\epsilon$-optimal policy can be achieved by 
Algorithm~\ref{alg:value_it} 
if the iteration stops with condition 
\begin{displaymath}
 || V_{\tau}(s) - V_{\tau-1}(s)||_{\infty} \le \frac{1-\rho}{2\rho} \epsilon 
\end{displaymath}
for all $s \in \mathbf S$. 
 \label{th:stop}
\end{theorem}
\begin{proof}
See Appendix~\ref{app:th:stop}.
\end{proof}

It can also be shown in Theorem~\ref{thm:conv} that Algorithm~\ref{alg:value_it} converges to the optimal policy $\pi^*$ by setting $\epsilon=0$. 
\begin{theorem} 
  \label{thm:conv}
  The optimal policy $\pi^*$ can always be achieved  by setting $\epsilon = 0$ in 
  Algorithm~\ref{alg:value_it}, i.e., 
  \begin{equation}
\lim_{\tau \rightarrow \infty}  V_{\tau}(s)  = V^*(s) 
\end{equation}
for all $s \in \mathbf S$.
\label{th:value_convergence}
\end{theorem}
\begin{proof}
See Appendix~\ref{app:th:value_convergence}.
\end{proof}
The proof of Theorem~\ref{th:value_convergence} is shown in \eqref{eqn:th_proof3}
in Appendix~\ref{app:th:value_convergence}, 
Based on this proof, we  concluded that 
\begin{align}
&\lim_{\tau \rightarrow \infty}||V_{\tau}(s) - V^*(s)||_{\infty}\le \lim_{\tau \rightarrow \infty} \rho^{\tau} ||V_0(s) - V^*(s)||_{\infty}\notag
\end{align}
where $\| \|_\infty$ denotes the infinite norm. This shows that 
the convergence speed of Algorithm~\ref{alg:value_it} 
significantly depends on the discount factor $\rho$. 
Hence, the convergence speed can be controlled by discount factor $\rho$.

Using the optimal policy, a node now can adaptively change its transmission range
against network dynamics, which leads to a robust network. 
It is worth noting that the complexity to find the optimal policy at each node does not change,  even if the total number of nodes in network increases. 
Hence, as the number of nodes increases, 
the total complexity to find the optimal policies of all nodes in the network increases linearly.   
This is because the proposed MDP framework of each node is not affected by individual network member nodes, instead it is only affected by node density $\lambda$ in network.  
In the next section, we
study how to determine the initial state for each node which determines the initial
transmission range.

\subsection{Stationary Network with Optimal Policy}
\label{subsec:equ}

Each node can 
periodically change its transmission range according to the optimal policy obtained from Algorithm~\ref{alg:value_it}. 
The resulting network can be in stationary, i.e., 
the number of network nodes is 
unchanged if each node takes action based on the optimal policy. 
In this section, we discuss how to initial conditions are determined such that  
the convergence speed for the optimal policy can be expedited in practice. 

With the optimal policy $\pi^*$, the proposed MDP framework is reduced to the Markov
chain with  a 
state transition matrix $\mathbf P $ whose element at $(s, s')$ is denoted
by $\mathbf P (s, s')$, which is expressed as 
\begin{align}
\mathbf P (s, s') 
&=P(s'|s,\pi^*(s))
\label{eqn:STM}
\end{align}
This is the state transition probability $P(s'|s,a)$ in~\eqref{eqn:stp_proof_re}
with the optimal action $a = \pi^*(s)$. 
The state transition probability $\mathbf P(s,s')$ provides the probability that a
single state transition changes  a node in $s$ to $s'$. 
Then, the 
limiting matrix $\lim_{n \rightarrow
\infty} \mathbf P^n$ and the limiting  distribution $\boldsymbol{\sigma}=
[\sigma_1,\dots,\sigma_s,
\ldots, \sigma_{|\mathbf S|}]$ 
which can be obtained  as 
\begin{align}
\sigma_s = \lim_{n \rightarrow \infty} \frac{\sum_{s' \in \mathbf S} \mathbf P^n (s',s)}{\sum_{s \in \mathbf S} \sum_{s' \in \mathbf S} \mathbf P^n (s',s)}
\label{eqn:sigma_s}
\end{align}
where $\sigma_s$ denotes the probability of being in
state $s$ after an infinite number of state transitions. 
Finally, the initial state $s^\dagger$ can be determined by choosing the state with
the
highest limiting distribution, i.e., 
\begin{align}
 s^\dagger 
 &=\arg_{s\in \mathbf S}\max \sigma_s 
%
\label{eqn:ini_first}
 \end{align}
which allows the initial network to be formed close to the stationary network with
the highest probability.  

\subsubsection{The optimal action includes no change of transmission range}
If the optimal action at a state $s^*$ is not to change its transmission range, 
i.e., $a^*=\pi^*(s^*)=0$, the $s^*$th row of $\mathbf P$ can be expressed  
by the definition of $P(s'|s,a)$ in \eqref{eqn:stp_proof_re} as
\begin{align}
\mathbf P (s^{*}, s') =
\begin{cases}
1, &  s' =s^{*}\\
0, & s' \in \{ \mathbf S \backslash s^{*} \}
\end{cases}
\label{eqn:absorbing}
\end{align}
where $ \mathbf S \backslash s^{*}$ denotes the set of elements in $\mathbf S $
excluding $s^{*}$. 
This allows the state transition matrix $\mathbf P$ to be formulated in canonical form as
\begin{align}
\mathbf P = \left(
\begin{array}{c c}
\mathbf Q &\mathbf R\\
\mathbf 0 &\mathbf I 
\end{array} \right)
\label{eqn:canonical_P_abs}
\end{align}
where $\mathbf Q\ge0$ is a nonnegative matrix, $\mathbf R>0$ is a strictly positive
matrix, $\mathbf 0$ denotes the matrix with zeros and $\mathbf I$ denotes the identity
matrix. 
The size of matrix $\mathbf I$ becomes the number of  states whose actions are zero.

Then, the limiting matrix of $\mathbf P$ in \eqref{eqn:canonical_P_abs} becomes
\begin{align}
\lim_{n \rightarrow \infty} 
\mathbf P^n &= 
\lim_{n \rightarrow \infty} 
 \left(
\begin{array}{c c}
\mathbf Q^n & \mathbf Q^{n-1} \mathbf {R}+\cdots +\mathbf Q \mathbf {R}+ \mathbf {R}\\
\mathbf 0 &\mathbf I 
\end{array} \right) \notag\\
&= 
 \left(
\begin{array}{c c}
\mathbf 0 & \mathbf {FR}\\
\mathbf 0 &\mathbf I 
\end{array} \right) 
\label{eqn:lim_P_2}
\end{align}
where $\mathbf F = (\mathbf I - \mathbf Q)^{-1}$ is the fundamental matrix of
$\mathbf Q$. $\lim_{n \rightarrow \infty} \mathbf Q^n = \mathbf 0$ as the
element of $\mathbf Q$ is in $[0,1)$. 
Then, an element $\sigma_s$ in the limiting distribution $\boldsymbol \sigma$ can be obtained based on \eqref{eqn:sigma_s}
\begin{align*}
\sigma_s = \frac{ \sum_{\forall i } {(\mathbf F \mathbf R)}_{ij}}{\zeta}
 \end{align*}
where ${(\mathbf F \mathbf R)}_{ij}$ denotes the $(i,j)$th element of the matrix $\mathbf{FR}$ and $\zeta = \sum_{\forall j }  \left( \sum_{\forall i } {(\mathbf F \mathbf R)}_{ij}+1\right)$ is a constant.  
Therefore, the initial state $s^\dagger$ can be determined as
\begin{align}
 s^\dagger &=\arg_{s\in \mathbf S}\max \sigma_s \notag\\
 &=\arg_{j}\max \sum_{\forall i} {(\mathbf F \mathbf R)}_{ij}. 
 \label{eqn:inistate_1}
\end{align}

In the case where the set of optimal actions does not include no change in transmission
range, 
the state transition matrix
$\mathbf P$ cannot be formulated as shown in \eqref{eqn:canonical_P_abs}.  
This is discussed next. 

\subsubsection{Optimal action includes  change of transmission range}
\label{subsubsec:opt_act}

Since $a^*=\pi^*(s^*)=0$ is not available as an action, the optimal actions are
either to enlarge or reduce the transmission range. 

If the optimal action at a state $s$ is to enlarge the transmission range (i.e., $a^*=\pi^*(s^*)>0$),
the $s$th row of $\mathbf P$ becomes  
\begin{align}
\mathbf P (s, s') =
\begin{cases}
0, &  s' < s\\
\frac{(\lambda a)^{\xi'-\xi} \cdot e^{-\lambda a}}{(\xi'-\xi)!} , & s' \ge s 
\end{cases}. 
\label{eqn:positive_a}
\end{align}
Similarly, 
if the optimal action at a state $s$ is to reduce the transmission range (i.e.,
$a^*=\pi^*(s^*)<0$), the $s$th row of $\mathbf P$ becomes  
\begin{align}
\mathbf P (s, s') =
\begin{cases}
{\xi\choose{\xi'} }{(1- \frac{|a|}{\bar a_{ \tau}})}^{\xi'}{( \frac{|a|}{\bar a_{ \tau}})}^{\xi- \xi'}, &  s' \le s\\
0, & s' > s 
\end{cases}. 
\label{eqn:negative_a}
\end{align}
The state transition matrix $\mathbf P $ can be correspondingly expressed as a canonical form of 
\begin{align}
\mathbf P = \left(
\begin{array}{c c}
\mathbf U &\mathbf Q_1\\
\mathbf Q_2 &\mathbf L 
\end{array} \right)
\label{eqn:canonical_P}
\end{align}
where $\mathbf U$ and $\mathbf L$  denote the upper and lower triangular matrices
and $\mathbf Q_1$ and $\mathbf Q_2$ are strictly positive matrices. 
Since the
optimal actions are determined by considering both rewards and costs, 
an optimal action can be determined to enlarge the current transmission range if a node
is in a state with too few nodes. 
On the other hand, if a node is in a state with too many nodes, the optimal policy
may determine the  optimal action that 
reduces the transmission range such that the cost can be reduced. 
Hence, the state transition matrix $\mathbf P$ in \eqref{eqn:canonical_P} 
consists of $(\mathbf U \quad \mathbf
Q_1)$ and  $(\mathbf Q_2 \quad \mathbf L)$.

Note that 
$\mathbf P^n$ for $n \ge 2$ is a strictly positive matrix. For example,
\begin{align*}
\mathbf P^2 &= \left(
\begin{array}{c c}
\mathbf U &\mathbf Q_1\\
\mathbf Q_2 &\mathbf L 
\end{array} \right)
\left(
\begin{array}{c c}
\mathbf U &\mathbf Q_1\\
\mathbf Q_2 &\mathbf L 
\end{array} \right)\\
&=\left(
\begin{array}{c c}
\mathbf U^2 +\mathbf Q_1\mathbf Q_2 & \mathbf U \mathbf Q_1 + \mathbf Q_2 \mathbf L\\
\mathbf Q_2 \mathbf U + \mathbf L \mathbf Q_2 &\mathbf Q_2\mathbf Q_1 +\mathbf L^2
\end{array} \right)
\end{align*}
which becomes a strictly positive matrix. 
Hence,  Perron-Frobenius theorem \cite{pillai2005} guarantees 
that 
there is a unique largest eigenvalue 
and the largest
eigenvalue is $1$ since $\mathbf P$ is a stochastic matrix. 
Therefore, the unique limiting distribution $\boldsymbol \sigma$ can be found as a 
row eigenvector of $\mathbf P$ associated with eigenvalue $1$, i.e., 
$\boldsymbol \sigma \mathbf P = \boldsymbol \sigma$, and the initial state becomes
the state with the largest limiting distribution as shown in \eqref{eqn:ini_first}.

The initialization phase of the proposed system can be expedited by  choosing the initial state of each node that leads to the initial network. 


\floatname{algorithm}{Algorithm}
\algsetup{indent= 1em}
\begin{algorithm}[tb]
        \caption{Algorithm for adaptation phase with packet transition }
        \label{alg:adapt}
\begin{algorithmic}[1]
\smallskip
\REQUIRE{ 
optimal policy $\pi^*$}
\STATE \textbf{Initialize:} $s \leftarrow s^\dagger$ \texttt{// build a stationary network}
\WHILE {network is active}
\STATE {\texttt{// receive and combine packets} \\
store received packets in buffer}
\IF{$\mathcal L_i \ne \emptyset$}
\STATE {{\hspace{-1em}}\texttt{//if the buffer is not empty} \\ 
build a network coded packet based on \eqref{eqn:NC1}}
\STATE {{\texttt{// update network topology}\\
check the current state $s$}
\STATE find the optimal action: $a^* \leftarrow \pi^*(s)$}
\STATE {update the transmission range: $\bar a_{\tau+1} \leftarrow \bar a_{\tau}+ a$}
\STATE {broadcast the network coded packet}
\ENDIF
\STATE {$\tau \leftarrow \tau+1$}
\ENDWHILE
\end{algorithmic}
\end{algorithm}

%
%
%

%

\section{Simulation Results}
\label{sec:simulation}

In this simulation, we consider a wireless ad hoc network with multi-source multicast
flows where multiple intermediate nodes aim to relay source data to multiple terminal
nodes using network coding. 
All intermediate nodes are policy-compliant, meaning that 
each node builds its own
optimal policy and correspondingly changes its transmission range
based on
the number of nodes included in its transmission range. 
In this section, we present a network formation result based on the proposed
strategy, and then we show a performance comparison with other existing network
formation strategies in applications with Wi-Fi Direct. 
The proposed algorithms are designed and implemented by MATLAB and all the
simulations are performed on a Windows 7 system configured by a Core i7 3.40GHz CPU with 8GB of RAM.

\subsection{Numerical Results for the Proposed Strategy}

We consider a network with two source nodes, two terminal
nodes and multiple intermediate nodes. The number of intermediate nodes follows
the PPP with a node density of $4/5$. The network size denotes 
the size of the area in the network,  
and 
three different network sizes are considered. 
The results presented in this section are based on $1,000$ independent 
experiments with a randomly generated number of nodes in a network size. 

Fig.~\ref{fig:area_numnodes} shows the number of nodes in given
network areas that 
are determined by the PPP with a node
density  of $\lambda = 4/5$.
\begin{figure}[tb]
\vspace{-0.5cm}
\centering
\begin{center}
\includegraphics[width = 8.5cm]{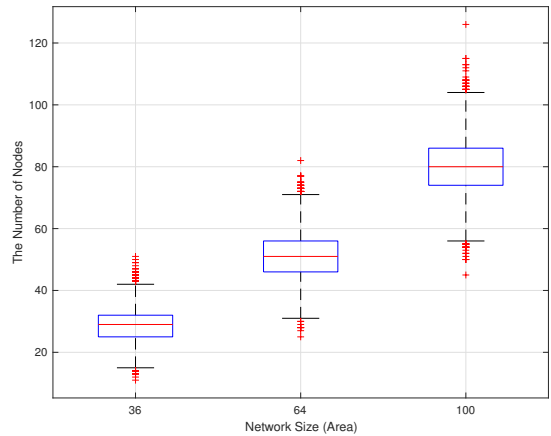}
\caption{ The number of appeared nodes for given network size with a node density of $\lambda = 4/5$.}
\label{fig:area_numnodes}
\end{center}
\vspace{-0.5cm}
\end{figure}
The line in the middle of each box in Fig.~\ref{fig:area_numnodes}
denotes the median of the experiments, which are $29$, $51$, and $80$ for 
network sizes of $36$, $64$, and
$100$, respectively.  The top and bottom of each box are the $25$th and $75$th
percentiles, respectively. Hence, it is confirmed that intermediate
nodes are well generated by the PPP based on the node density.

Each agent builds $\epsilon$-optimal policy based on Algorithm~\ref{alg:value_it}
with the parameter $\epsilon =
0.01$, twenty states and five actions, i.e., $|\mathbf S|
= 20$ and $|\mathbf A|=5$. 
\begin{figure}[tb]
\centering
\begin{center}
\includegraphics[width = 8.5cm]{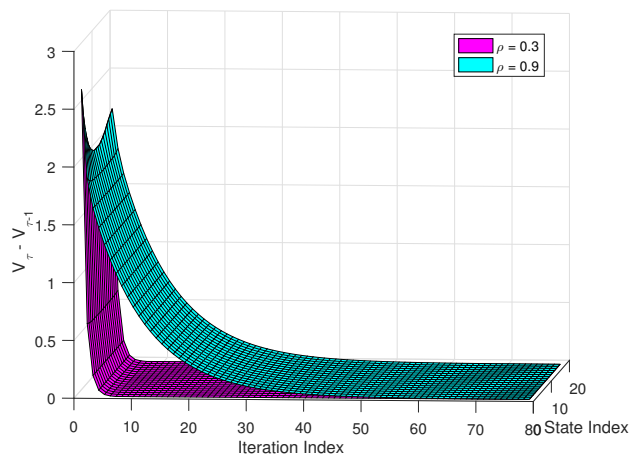}
\caption{  Convergence process in Algorithm~\ref{alg:value_it}}
\label{fig:3D_conv}
\end{center}
\vspace{-1cm}
\end{figure}
Fig.~\ref{fig:3D_conv} shows the values of 
$V_{\tau}(s) - V_{\tau-1}(s)$ for all $s \in \mathbf S$
over iterations and it is observed that 
$V_{\tau}(s) - V_{\tau-1}(s)$ approaches 0 as the number of iterations increases. 
Specifically, 
the iteration is terminated if 
$ V_{\tau}(s) - V_{\tau-1}(s) \le \frac{1-\rho}{2\rho} \epsilon$,  
and thus eventually  
$ V_{\tau}(s) = V_{\tau-1}(s) $ with large enough iterations.
Therefore, we conclude that the proposed algorithm converges.  
The convergence speed is dependent on the 
discount factor $\rho$ as shown in Fig.~\ref{fig:rho_it}.  
For larger $\rho$, which takes into account longer future utilities, 
it takes a longer time (i.e., more iterations) to find $\epsilon$-policy.  
On the other hand, it takes less time to find the $\epsilon$-policy for a small $\rho$. 
\begin{figure}[t]
\centering
\begin{center}
\includegraphics[width = 8cm]{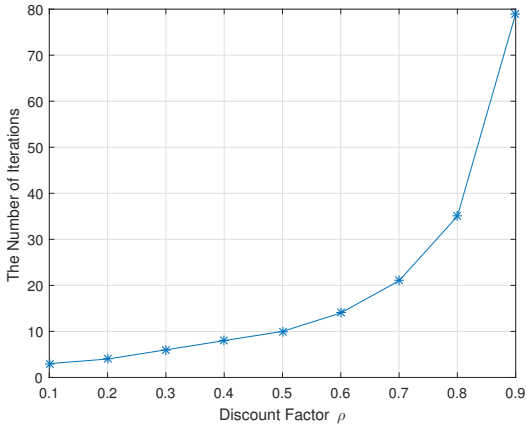}
\caption{ Convergence speed of Algorithm~\ref{alg:value_it} as a function of discount factor}
\label{fig:rho_it}
\end{center}
\vspace{-1.5cm}
\end{figure}

%
We next study the resulting network in terms of two connectivity measures: 
the number of constructed links 
and algebraic connectivity~\cite{gross2004}.
The number of links constructed in the network reflects 
the extrinsic connectivity and can be quantified by counting the number
of links in the network. In contrast,  the algebraic connectivity 
is the measure of intrinsic
connectivity, i.e., 
how well the overall network is constructed. 
%
Fig.~\ref{fig:zero_beta} shows the impact of weight $\omega$ in utility
function \eqref{eqn:U} on network connectivity. Since $\omega$ is the weight of
reward in the utility function, it is expected that the resulting networks are formed
such that the rewards (or the cost) are given more weight 
than the cost (or the rewards)
if $\omega$ is high (or low). 
In the simulations, we assume that there is no link failure in the
channels (i.e., $\beta = 0$) and the discount factor is $\rho = 0.5$. 
Fig.~\ref{fig:area_links} shows that the number of links is proportional to
both network size and $\omega$. 
\begin{figure*}[tb]
    \centering
    \begin{subfigure}[b]{0.4\textwidth}
        \includegraphics[width=\textwidth]{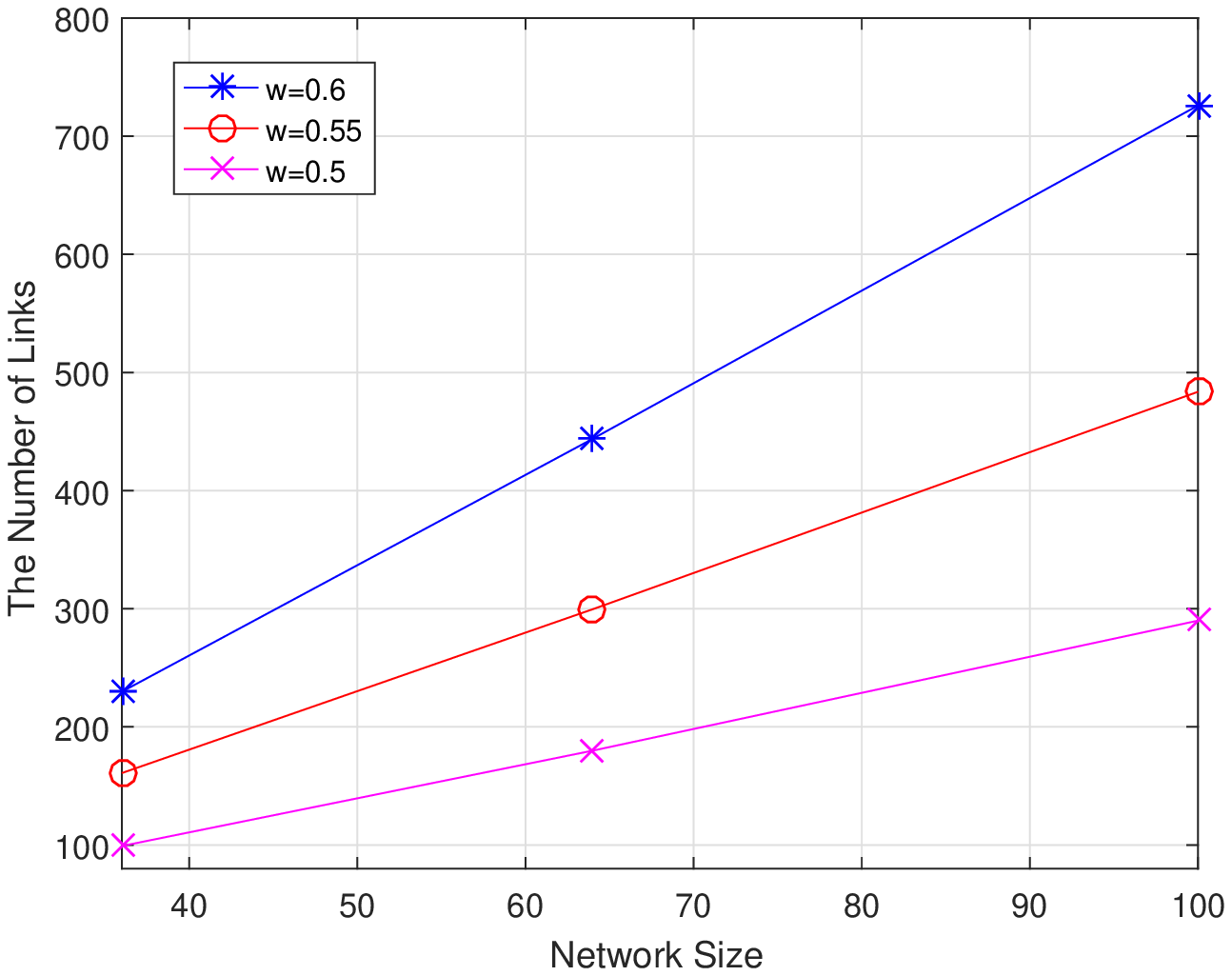}
        \caption{}
        \label{fig:area_links}
    \end{subfigure}
    ~ 
    \begin{subfigure}[b]{0.4\textwidth}
        \includegraphics[width=\textwidth]{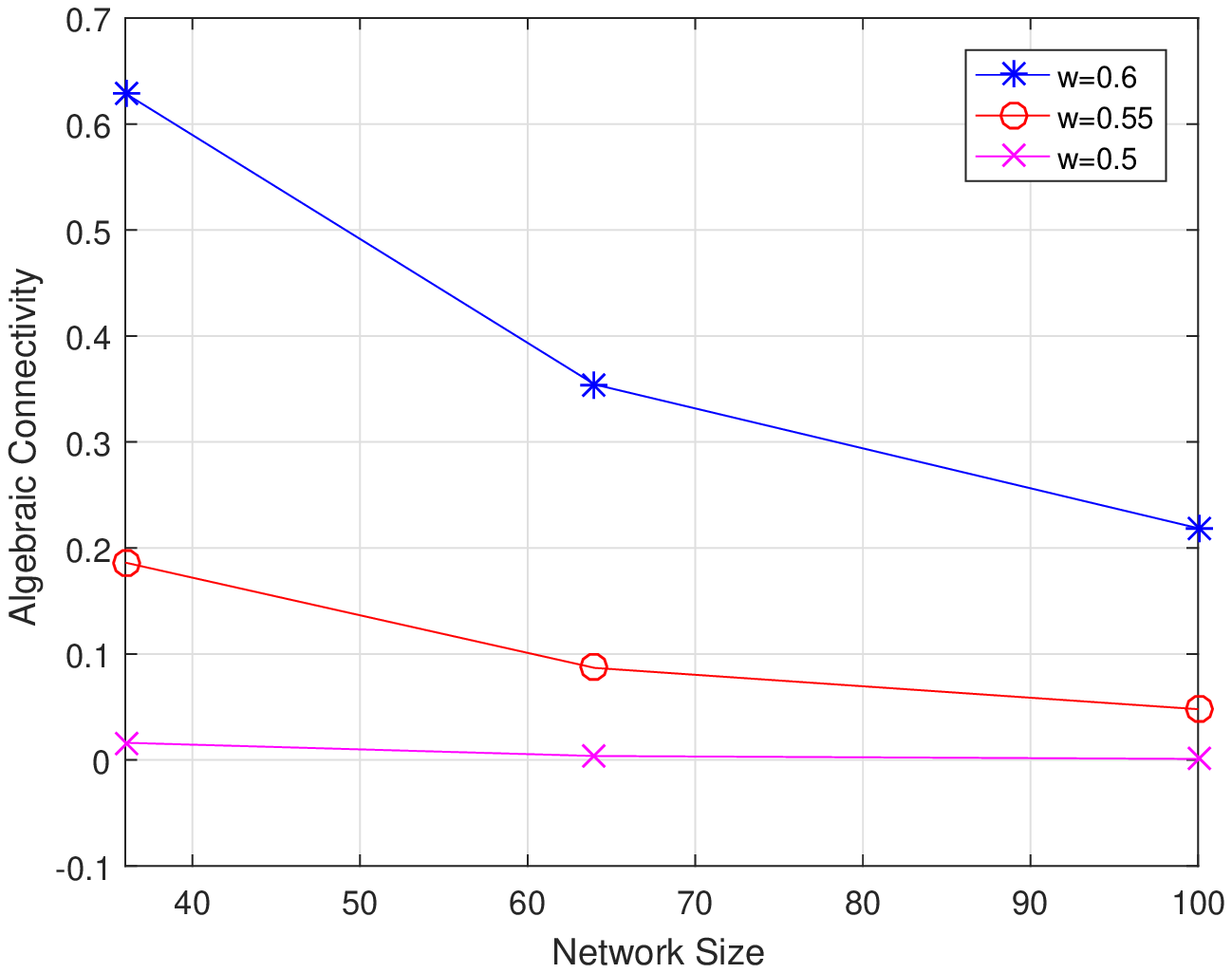}
        \caption{}
        \label{fig:area_alg}
    \end{subfigure}
    \caption{The impact of $\omega$ on network connectivity ($\beta = 0, \rho = 0.5$)}
    \label{fig:zero_beta}
    \vspace{-0.5cm}
\end{figure*}
A node with high $\omega$ may increase the 
transmission range such that a larger 
number
of links can be covered, leading to  
throughput gain over power
consumption. 
This is also confirmed in Fig.~\ref{fig:area_alg}, which shows high algebraic
connectivity with high $\omega$. 
However, 
Fig.~\ref{fig:area_alg} shows that the
algebraic connectivity decreases as network size increases. 
This is because the proposed strategy does not consider to retain 
the same algebraic connectivity.
Hence, if the same algebraic connectivity is required, a higher $\omega$
should be considered in a larger network. 


\begin{figure*}[tb]
    \centering
    \begin{subfigure}[b]{0.4\textwidth}
        \includegraphics[width=\textwidth]{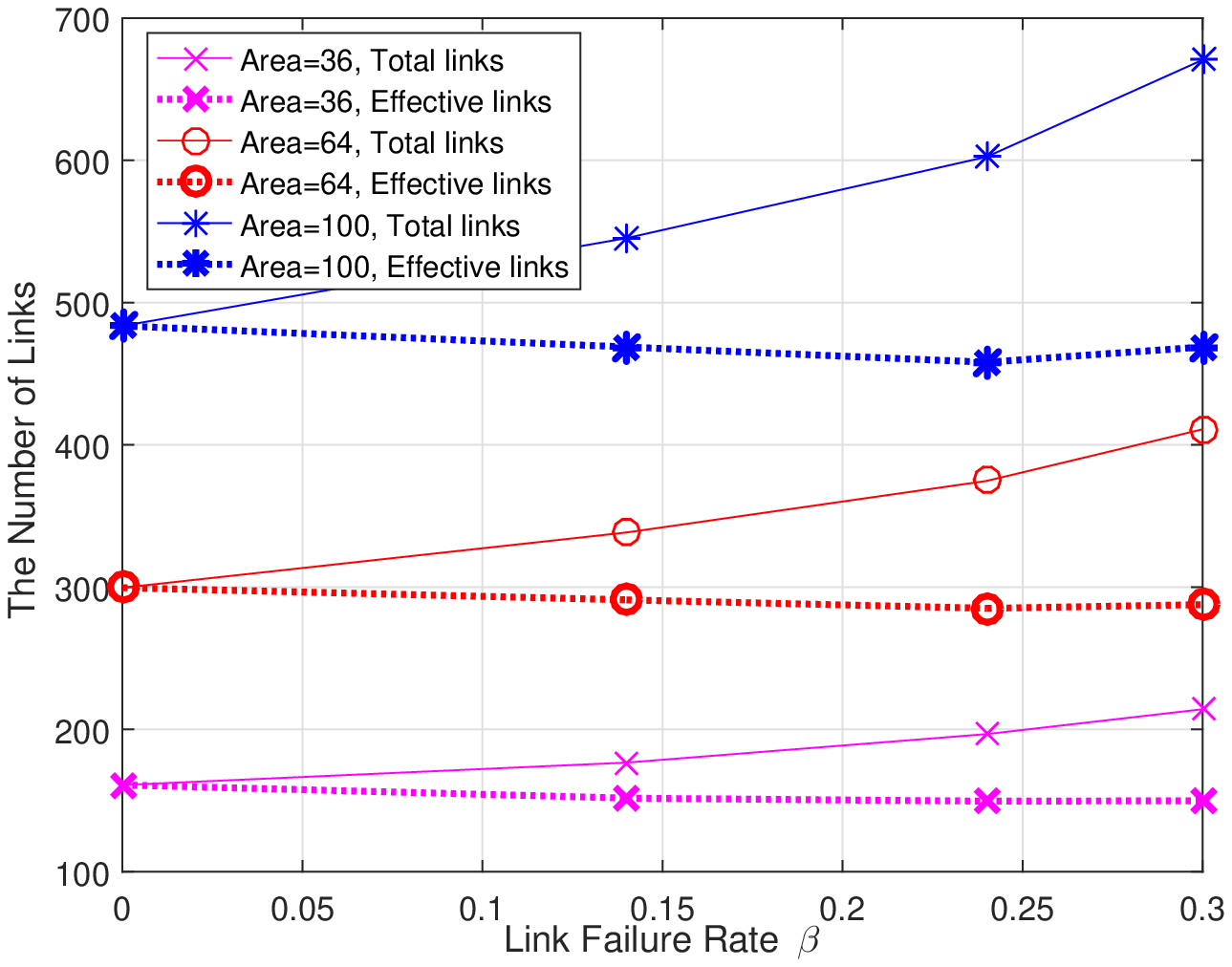}
        \caption{}
        \label{fig:beta_links}
    \end{subfigure}
    ~ 
    \begin{subfigure}[b]{0.4\textwidth}
        \includegraphics[width=\textwidth]{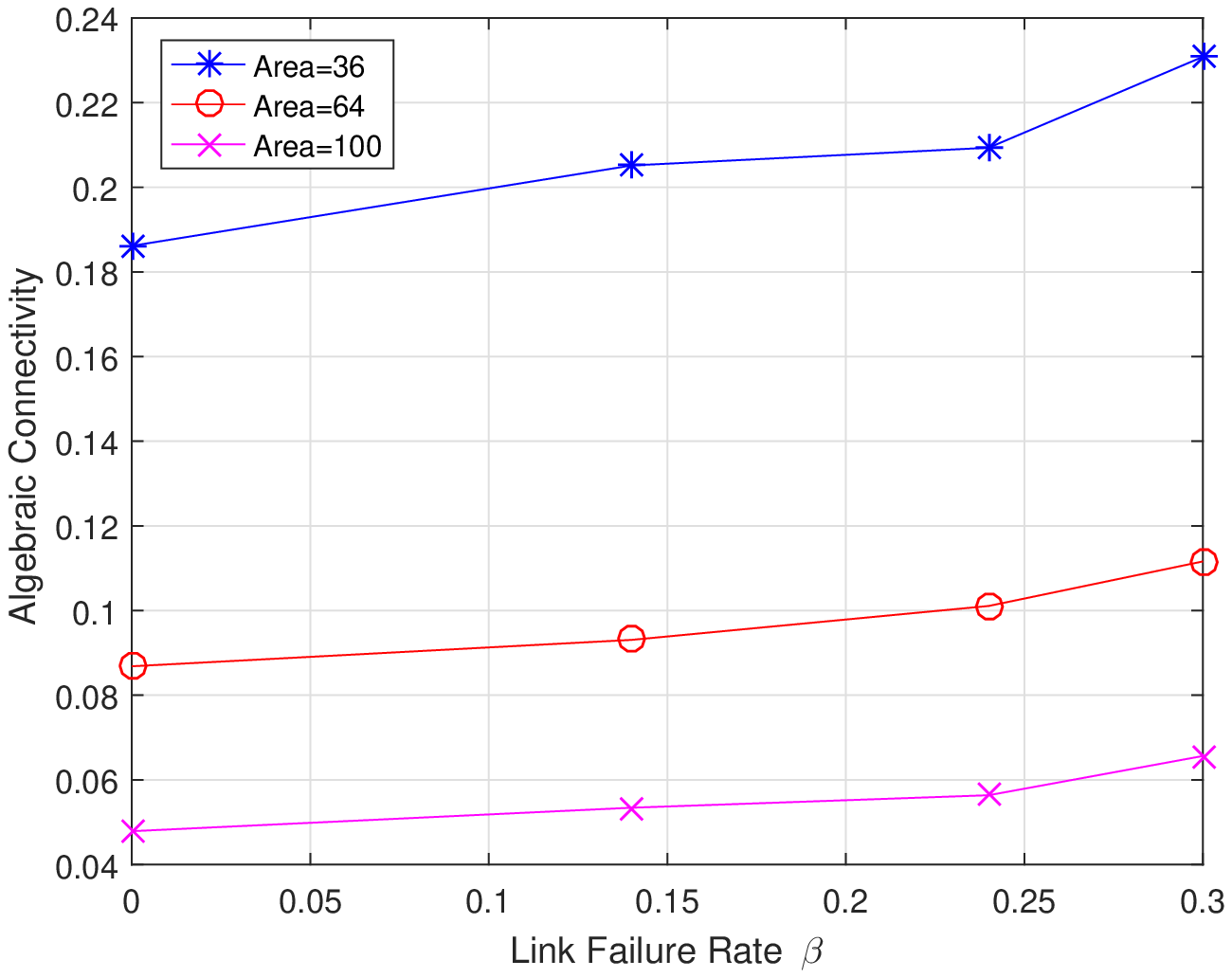}
        \caption{}
        \label{fig:beta_alg}
    \end{subfigure}
    \caption{Network connectivity over link failure rate ($0\le \beta \le 0.3, \omega  = 0.55, \rho = 0.5$)
    }
    \label{fig:betas}
    \vspace{-0.5cm}
\end{figure*}

%
%
The network connectivity as a function of the link failure rate $\beta$ ($0 \le \beta
\le 0.3$) is shown in Fig.~\ref{fig:betas}. 
Fig.~\ref{fig:beta_links} shows that the proposed strategy enables
nodes to make more
links as $\beta$ increases. This enables the networks to maintain  
approximately the same number of effective nodes. 
Moreover, it is confirmed from Fig.~\ref{fig:beta_alg} that 
the algebraic connectivity increases as
$\beta$ increases.  
This is because the proposed approach increases the degree of connectivity 
of the network to overcome
unstable channel conditions. 
Therefore, we conclude that the proposed
strategy is successful at adaptively changing network topology by explicitly considering the 
link failure rates of the channels. 


%
%


\subsection{Performance Comparison in Wi-Fi Direct Application} 

In this section, we consider an illustrative application with
Wi-Fi Direct where data are transmitted over dynamic wireless networks 
in a $60\times 60$ [m$^2$] area.  
Mobile nodes are located at a  density of $8\times 10^{-3}$[nodes/m$^2$] and are connected by 
Wi-Fi Direct with IEEE 802.11ac standard MCS-9. 
The parameters used in the simulations are shown in
Table~\ref{table:simulation_parameter}, and they are specified by the 
IEEE 802.11ac standard~\cite{IEEE80211ac, cisco80211ac}. 
The RLNC is used in the GF($2^8$).  
The performance of the proposed strategy  
is evaluated based on 
the system goodput~\cite{miao2016}, which is defined as the
sum of data rates successfully delivered to terminal nodes, expressed as 
\begin{equation*} 
 \sum_{h = 1}^{N_H}\sum_{v_t \in \mathbf {\tilde T}_h} \frac{ L }{ \bar \tau(x_h,v_t)}
\end{equation*}
where $\mathbf {\tilde T}_h \subseteq \mathbf T_h$ denotes a set of successfully delivered terminal nodes of $x_h$,
 $L $ represents the size  data set $x_h$,
and $\bar \tau({x_h,v_t})$ denotes the travel time\footnote{The travel time refers to the time taken for a data set to be transmitted across a network from source to terminal node. It includes processing delay, transmission delay, propagation delay and queuing delay. }
 for data set $x_h$ to  arrive to a
terminal node $v_t \in \mathbf {\tilde T}_h$. 
Moreover, the transmission power is measured 
by a path loss model, expressed as 
\begin{displaymath}
 P_{TX}= P_{RX} \cdot {\left(\frac{4\pi}{\lambda}\cdot d\right)}^\alpha =\eta \cdot
d^\alpha 
\end{displaymath}
where $P_{TX}$, $P_{RX}$, $\lambda$, and $d$ denote  
transmission power, receive power,  wave length, and the 
distance between transmitter and receiver, respectively. 
\begin{table}[tb]
\caption{Simulation Parameters}
\label{table:simulation_parameter}
\centering
\begin{tabular}{|l|l||l|l|}
\hline
Parameter & Value & Parameter & Value\\
\hline \hline
$\eta$&$1$&$\alpha$&2  \\ \hline
Channel Bandwidth&80 $MHz$&TX-RX Antennas&$3 \times 3$\\ \hline
Modulation Type&256-QAM &Coding Rate&5/6\\ \hline
Guard Interval &400$ns$&PHY Data Rate&1300 $Mbps$\\ \hline
MAC Efficiency&70$\% $&Throughput&910 $Mbps$\\ \hline
\end{tabular}
\vspace{-0.6cm}
\end{table} 

In this simulation, we simultaneously consider three types of network dynamics:
changes in  member nodes of considered network, link failure rates, and node locations. 
To produce realistic dynamic network settings, the location of network nodes is changed in every time stamp, and the network member is updated, and the link failure rates $\beta$ are updated (i.e., randomly selected in $[0,0.3]$) every $5$ time stamps. 
The simulation parameters are set to $\omega =0.53$, $u=0.2$, $\epsilon = 0.01$, $|\mathbf S| = 18$ and $|\mathbf A|=7$. 

We compare the performance of the proposed strategy with the following three existing 
network formation strategies.

\begin{enumerate}
\item  \emph{Myopic}: A myopic strategy is a special case of 
  the proposed strategy with the setting of $\rho=0$ 
  in \eqref{eqn:value}. The myopic solution does not consider the future utilities.
  Rather, it focuses on maximizing the immediate utility only, i.e., 
\begin{align*}
\pi^{myop}  
= \arg\max_{a \in \mathbf A} \sum_{s' \in \mathbf S}P(s'|s,a)  U(s,a,s')  , \forall s \in \mathbf S.   
\end{align*} 

\item \emph{Traskov}\cite{traskov2006}: A well-known centralized network formation
 strategy for network coding deployed networks. 
Traskov can provide a static network topology for a given node distribution by 
exploiting network coding opportunities.
Hence, in the simulations, we 
consider the network where $(\text{network size}
\times \lambda)$ nodes are uniformly distributed, and we find the network topology based on Traskov. 
Since Traskov determines individual links,  the transmission range of a
node is assigned to include all links determined from Traskov. 
To ensure a fair comparison, the assigned transmission range is not changed over time since the  computational complexity for Traskov is much higher than that of other distributed strategies.  

%

\item \emph{TCLE}\cite{xu2016}: A state-of-the-art  distributed strategy  for network
  formation based
  on a non-cooperative game.  
In this strategy, a node chooses its transmission power by balancing the target
algebraic 
connectivity against transmission energy dissipation. 
To ensure a fair comparison, the same set of actions are 
employed as the proposed strategy 
and 
the target algebraic connectivity of TCLE is set as $0.1$, which is 
the average algebraic
connectivity of the proposed strategy with $\omega=0.53$.
Note that the network topology determined by TCLE 
does not adaptively change  


To combat network dynamics, we allow TCLE to recalculate its solution every $5$ time stamps.  
Note that in terms of computational complexity, these TCLE settings require higher complexity than the proposed and myopic strategies, where nodes simply lookup the optimal policies during all
the simulations, which are obtained 
in the beginning of the simulations. 

\end{enumerate}


\begin{table*}[tb]
\caption{Numerical Results of Wi-Fi Direct Application}
\label{tab:results}
\renewcommand{\arraystretch}{1}
\begin{center}
\begin{tabular}{  | c|| c| c|c|  }
\hline
 Strategy &System Goodput $[Mbps]$&Successful Connectivity Ratio $[\%]$& Power Consumption $[dBm]$\\
 \hline
\hline
Proposed& 324.60 & 75.63 & 89.07\\
 \hline
Myopic& 276.51 & 70.28 & 80.17\\
 \hline
Traskov & 226.59 & 38.19 & 103.14 \\
 \hline
TCLE& 317.36 & 45.68& 76.22 \\
\hline 
\end{tabular}
\end{center}
\vspace{-0.8cm}
\end{table*}

The average numerical results from $1,000$ time stamps are summarized in
Table~\ref{tab:results}, and 
illustrative results in the time stamp range of $[380,440]$  are shown 
in Fig.~\ref{fig:tx_radius} and Fig.~\ref{fig:tot_s} for 
 the radius of transmission range of a node and the number of total links in 
the overall network, respectively. 

\begin{figure}[tb]
\centering
\begin{center}
\includegraphics[width = 8.5cm]{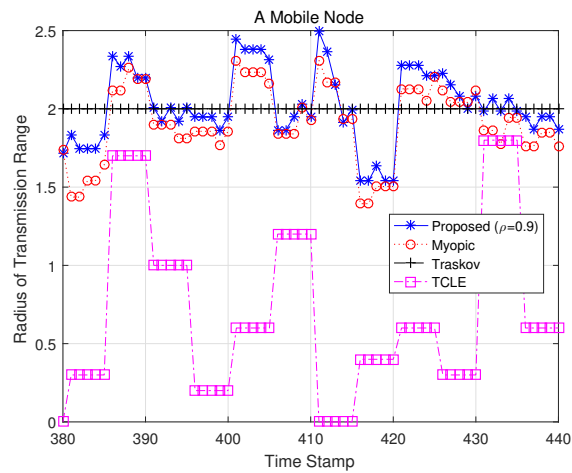}
\caption{ Radius of transmission range of a node in presence of network dynamics.}
\label{fig:tx_radius}
\end{center}
\vspace{-0.8cm}
\end{figure}
\begin{figure}[tb]
\centering
\begin{center}
\includegraphics[width = 8.5cm]{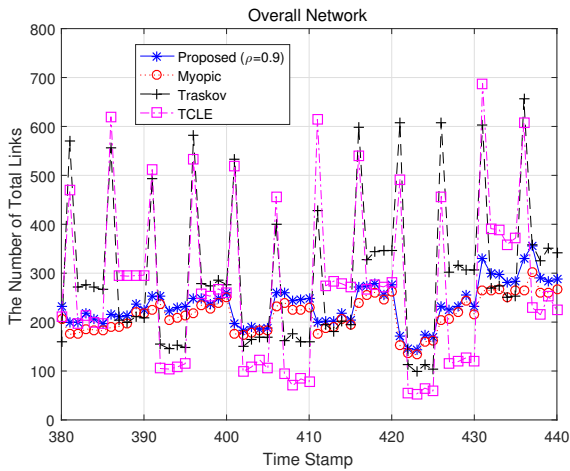}
\caption{ The number of links in the resulting network in the presence of network dynamics.}
\label{fig:tot_s}
\end{center}
\vspace{-0.8cm}
\end{figure}
As shown in  Table~\ref{tab:results}, the proposed strategy provides the highest system
goodput as well as high successful connectivity ratio. 
The proposed strategy always outperforms the myopic strategy. 
This is
because the policy of the myopic strategy focuses only on the immediate utility, while 
the proposed strategy  considers not only the immediate utility but also 
future utilities. 
For example,  time stamps in $[430, 440]$ 
of Fig.~\ref{fig:tot_s} show  
that the proposed strategy more proactively responds to network dynamics than the 
myopic strategy by changing larger number of network links.    
Moreover, it is confirmed that the myopic strategy tends to result in smaller
transmission ranges (shown in Fig.~\ref{fig:tx_radius}), which leads to a lower
number of total 
active links
in the network (shown in Fig.~\ref{fig:tot_s}). 
Since small transmission range requires lower transmission power consumption, the myopic strategy consumes lower power compared to the proposed strategy. However, in terms of \emph {power efficiency} which can be computed as  system goodput per unit power $[Mbps/dBm]$ from Table~\ref{tab:results}, 
the proposed strategy has $3.644$ which is higher than  myopic strategy that is $3.449$, leading to improved system goodput given power budgets. 

While the second highest system goodput is achieved by the TCLE, 
it shows the second lowest successful
connectivity ratio in Table~\ref{tab:results}. This implies that the TCLE can make
successful connections between a source and a terminal based on 
significantly 
short paths. 
However, the TCLE
is not an appropriate solution
for applications that count on successful delivery over throughput.
Rather, it is the most energy-efficient strategy
(Table~\ref{tab:results}) as highlighted in  \cite{xu2016}, and it can determine  smaller 
transmission ranges (in Fig.~\ref{fig:tx_radius}).

Traskov shows the lowest performance in terms of system goodput and successful
connectivity ratio while it requires the highest power consumption.
As shown in Fig.~\ref{fig:tx_radius}, Traskov does not change the transmission range
once it is  determined in the beginning of the simulation such that the result of network
formation fails to overcome network dynamics.

\section{Conclusions}
\label{sec:conclusion}
In this paper, we focus on a distributed network formation strategy that can build a 
robust network against network dynamics. 
We show that network coding induces packet anonymity and network decoupling such that
the MDP framework can be employed at each intermediate node. 
The intermediate nodes determine an optimal policy based on MDP, 
and the policy allows  the nodes to determine optimal 
transmission ranges that maximize the long-term cumulative utilities. The optimal
transmission ranges are determined by explicitly considering
current network conditions and future network dynamics. 
We further show that the resulting network of the proposed strategy converges to the
stationary networks,  
and we propose how to determine an initial network that can rapidly converge to the
stationary network. 
Simulation results confirm that the resulting network of
the proposed strategy can adaptively change by responding to network dynamics such as unstable channel condition with high link failure rate, 
node mobility, and corresponding changes in member nodes associated with the considered network. 

%


\bibliographystyle{IEEEtran}
\bibliography{IEEEabrv,mybibfile_v10}

\begin{thebibliography}{10}
\providecommand{\url}[1]{#1}
\csname url@samestyle\endcsname
\providecommand{\newblock}{\relax}
\providecommand{\bibinfo}[2]{#2}
\providecommand{\BIBentrySTDinterwordspacing}{\spaceskip=0pt\relax}
\providecommand{\BIBentryALTinterwordstretchfactor}{4}
\providecommand{\BIBentryALTinterwordspacing}{\spaceskip=\fontdimen2\font plus
\BIBentryALTinterwordstretchfactor\fontdimen3\font minus
  \fontdimen4\font\relax}
\providecommand{\BIBforeignlanguage}[2]{{%
\expandafter\ifx\csname l@#1\endcsname\relax
\typeout{** WARNING: IEEEtran.bst: No hyphenation pattern has been}%
\typeout{** loaded for the language `#1'. Using the pattern for}%
\typeout{** the default language instead.}%
\else
\language=\csname l@#1\endcsname
\fi
#2}}
\providecommand{\BIBdecl}{\relax}
\BIBdecl

\bibitem{Ahlswede2000}
R.~Ahlswede, N.~Cai, S.-Y.~R. Li, and R.~W. Yeung, ``Network information
  flow,'' \emph{IEEE Transactions on Information Theory}, vol.~46, no.~4, pp.
  1204--1216, Jul. 2000.

\bibitem{wang2008}
C.-C. Wang and N.~B. Shroff, ``On wireless network scheduling with intersession
  network coding,'' in \emph{IEEE Annual Conference on Information Sciences and
  Systems}, 2008, pp. 30--35.

\bibitem{Kim2009jsac}
Y.~Kim and G.~D. Veciana, ``Is rate adaptation beneficial for inter-session
  network coding?'' \emph{IEEE Journal on Selected Areas in Communications},
  vol.~27, no.~5, pp. 635--646, Jun. 2009.

\bibitem{Khreishah2010}
A.~Khreishah, C.~C. Wang, and N.~B. Shroff, ``Rate control with pairwise
  intersession network coding,'' \emph{IEEE/ACM Transactions on Networking},
  vol.~18, no.~3, pp. 816--829, Jun. 2010.

\bibitem{Bourtsoulatze2014TCOM}
E.~Bourtsoulatze, N.~Thomos, and P.~Frossard, ``Decoding delay minimization in
  inter-session network coding,'' \emph{IEEE Transactions on Communications},
  vol.~62, no.~6, pp. 1944--1957, Jun. 2014.

\bibitem{Bourtsoulatze2014TMM}
------, ``Distributed rate allocation in inter-session network coding,''
  \emph{IEEE Transactions on Multimedia}, vol.~16, no.~6, pp. 1752--1765, Oct.
  2014.

\bibitem{Hulya2009}
Hulya and Athina, ``Distributed rate control for video streaming over wireless
  networks with intersession network coding,'' in \emph{2009 17th International
  Packet Video Workshop}, May 2009, pp. 1--10.

\bibitem{Douik2016}
A.~Douik, S.~Sorour, T.~Y. Al-Naffouri, and M.~S. Alouini, ``Decoding delay
  controlled completion time reduction in instantly decodable network coding,''
  \emph{IEEE Transactions on Vehicular Technology}, vol.~PP, no.~99, pp. 1--1,
  2016.

\bibitem{Xie2016}
``Virtual overhearing: An effective way to increase network coding
  opportunities in wireless ad-hoc networks,'' \emph{Computer Networks}, vol.
  105, pp. 111 -- 123, 2016.

\bibitem{Li2003}
S.-Y.~R. Li, R.~W. Yeung, and N.~Cai, ``Linear network coding,'' \emph{IEEE
  Transactions on Information Theory}, vol.~49, no.~2, pp. 371--381, Feb. 2003.

\bibitem{Ho2003}
T.~Ho, R.~Koetter, M.~M{\'e}dard, D.~Karger, and M.~Effros, ``The benefits of
  coding over routing in a randomized setting,'' in \emph{IEEE International
  Symposium on Information Theory}, Cambridge, MA, USA, Jun/Jul. 2003.

\bibitem{RandomizedNC03}
T.~Ho, M.~M\'{e}dard, J.~Shi, M.~Effros, and D.~R. Karger, ``On randomized
  network coding,'' in \emph{Annual Allerton Conference on Communication,
  Control, and Computing}, Monticello, IL, USA, Oct. 2003.

\bibitem{PracticalNC03}
P.~A. Chou, Y.~Wu, and K.~Jain, ``Practical network coding,'' in \emph{Annual
  Allerton Conference on Communication, Control, and Computing}, Monticell, IL,
  USA, Oct. 2003.

\bibitem{Koetter2008}
R.~Koetter and F.~R. Kschischang, ``Coding for errors and erasures in random
  network coding,'' \emph{IEEE Transactions on Information Theory}, vol.~54,
  no.~8, pp. 3579--3591, Aug 2008.

\bibitem{Zhang2008information}
Z.~Zhang, ``Linear network error correction codes in packet networks,''
  \emph{IEEE Transactions on Information Theory}, vol.~54, no.~1, pp. 209--218,
  Jan 2008.

\bibitem{Esmaeilzadeh2017}
M.~Esmaeilzadeh, P.~Sadeghi, and N.~Aboutorab, ``Random linear network coding
  for wireless layered video broadcast: General design methods for adaptive
  feedback-free transmission,'' \emph{IEEE Transactions on Communications},
  vol.~65, no.~2, pp. 790--805, Feb 2017.

\bibitem{Jain2003}
K.~Jain, M.~Mahdian, and M.~R. Salavatipour, ``Packing steiner trees,'' in
  \emph{ACM-SIAM symposium on Discrete algorithms}, 2003, pp. 266--274.

\bibitem{Cui2004}
Y.~Cui, Y.~Xue, and K.~Nahrstedt, ``Optimal distributed multicast routing using
  network coding: Theory and applications,'' \emph{ACM SIGMETRICS Performance
  Evaluation Review}, vol.~32, no.~2, pp. 47--49, Sep. 2004.

\bibitem{Lun2005}
D.~S. Lun, N.~Ratnakar, R.~Koetter, M.~M\'{e}dard, E.~Ahmed, and H.~Lee,
  ``Achieving minimum-cost multicast: a decentralized approach based on network
  coding,'' in \emph{IEEE Annual Joint Conference of the Computer and
  Communications Societies}, vol.~3, March 2005, pp. 1607--1617.

\bibitem{Yan2006}
X.~Yan, J.~Yang, and Z.~Zhang, ``An outer bound for multisource multisink
  network coding with minimum cost consideration,'' \emph{IEEE Transactions on
  Information Theory}, vol.~52, no.~6, pp. 2373--2385, Jun. 2006.

\bibitem{traskov2006}
D.~Traskov, N.~Ratnakar, D.~S. Lun, R.~Koetter, and M.~M\'{e}dard, ``Network
  coding for multiple unicasts: An approach based on linear optimization,'' in
  \emph{IEEE International Symposium on Information Theory}, July 2006, pp.
  1758--1762.

\bibitem{mhkwon2017SPL}
M.~Kwon and H.~Park, ``Distributed network formation strategy for network
  coding based wireless networks,'' \emph{IEEE Signal Processing Letters},
  vol.~24, no.~4, pp. 432--436, April 2017.

\bibitem{sohrabi2000}
K.~Sohrabi, J.~Gao, V.~Ailawadhi, and G.~J. Pottie, ``Protocols for
  self-organization of a wireless sensor network,'' \emph{IEEE Personal
  Communications}, vol.~7, no.~5, pp. 16--27, 2000.

\bibitem{Gao2016}
T.~Gao, F.~Lang, and N.~Guo, ``An emergency communication system based on
  uav-assisted self-organizing network,'' in \emph{International Conference on
  Innovative Mobile and Internet Services in Ubiquitous Computing}, July 2016,
  pp. 90--95.

\bibitem{xu2016}
M.~Xu, Q.~Yang, and K.~S. Kwak, ``Distributed topology control with lifetime
  extension based on non-cooperative game for wireless sensor networks,''
  \emph{IEEE Sensors Journal}, vol.~16, no.~9, pp. 3332--3342, 2016.

\bibitem{Andrews2010}
J.~G. Andrews, R.~K. Ganti, M.~Haenggi, N.~Jindal, and S.~Weber, ``A primer on
  spatial modeling and analysis in wireless networks,'' \emph{IEEE
  Communications Magazine}, vol.~48, no.~11, pp. 156--163, November 2010.

\bibitem{Huang2014}
K.~Huang and V.~K.~N. Lau, ``Enabling wireless power transfer in cellular
  networks: Architecture, modeling and deployment,'' \emph{IEEE Transactions on
  Wireless Communications}, vol.~13, no.~2, pp. 902--912, Feb. 2014.

\bibitem{chou2007}
P.~A. Chou and Y.~Wu, ``Network coding for the internet and wireless
  networks,'' \emph{IEEE Signal Processing Magazine}, vol.~24, no.~5, pp.
  77--85, Sep. 2007.

\bibitem{Katti2006}
S.~Katti, H.~Rahul, H.~Wenjun, D.~Katabi, M.~M{\'e}dard, and J.~Crowcroft,
  ``Xors in the air: practical wireless network coding,'' \emph{IEEE/ACM
  Transactions on Networking}, vol.~16, no.~3, pp. 497 --510, Jun. 2008.

\bibitem{nad2004}
T.~Nad and A.~Krishnamurthy, ``Problems with network coding in overlay
  networks,'' \emph{Techinical Report, Yale University}, 2004.

\bibitem{topakkaya2011}
H.~Topakkaya, ``Network coding for wireless and wired networks: Design,
  performance and achievable rates,'' Ph.D. dissertation, Iowa State
  University, 2011.

\bibitem{Cloud2012}
J.~Cloud, L.~M. Zeger, and M.~M\'{e}dard, ``Mac centered cooperation -
  synergistic design of network coding, multi-packet reception, and improved
  fairness to increase network throughput,'' \emph{IEEE Journal on Selected
  Areas in Communications}, vol.~30, no.~2, pp. 341--349, Feb. 2012.

\bibitem{mirrezaei2014}
S.~M. Mirrezaei, M.~Dosaranian-Moghadam, and M.~Yazdanpanahei, ``Effect of
  network coding and multi-packet reception on point-to-multi-point broadcast
  networks,'' \emph{Wireless Personal Communications}, vol.~79, no.~3, pp.
  1859--1891, 2014.

\bibitem{HO2006}
M.~M{\'e}dard, R.~Koetter, D.~Karger, M.~Effros, J.~Shi, and B.~Leong, ``A
  random linear network coding approach to multicast,'' \emph{IEEE Transactions
  on Information Theory}, vol.~52, no.~10, pp. 4413--4430, Oct 2006.

\bibitem{Katti2005}
S.~Katti, D.~Katabi, W.~Hu, H.~Rahul, and M.~M\'{e}dard, ``The importance of
  being opportunistic: Practical network coding for wireless environments,'' in
  \emph{Allerton Annual Conference on Communication}, 2005.

\bibitem{Liu2007infocom}
J.~Liu, D.~Goeckel, and D.~Towsley, ``Bounds on the gain of network coding and
  broadcasting in wireless networks,'' in \emph{IEEE International Conference
  on Computer Communications}, May 2007, pp. 724--732.

\bibitem{bathe1976}
K.-J. Bathe and E.~L. Wilson, ``Numerical methods in finite element analysis,''
  \emph{Prentice-Hall Englewood Cliffs, NJ}, 1976.

\bibitem{mhKWON2014WCNC}
M.~Kwon, H.~Park, and P.~Frossard, ``Compressed network coding: Overcome
  all-or-nothing problem in finite fields,'' in \emph{IEEE Wireless
  Communications and Networking Conference}, Apr. 2014, pp. 2851--2856.

\bibitem{Yan2013}
Z.~Yan, H.~Xie, and B.~W. Suter, ``Rank deficient decoding of linear network
  coding,'' in \emph{IEEE International Conference on Acoustics, Speech and
  Signal Processing}, May 2013, pp. 5080--5084.

\bibitem{mhKWON2016elsevier}
M.~Kwon, H.~Park, N.~Thomos, and P.~Frossard, ``Approximate decoding for
  network coded inter-dependent data,'' \emph{Signal Processing}, vol. 120, pp.
  222--235, Mar. 2016.

\bibitem{KWON2016}
M.~Kwon and H.~Park, ``The impact of network coding cluster sizes on the
  approximate decoding performance,'' \emph{KSII Transactions on Internet and
  Information Systems}, vol.~10, no.~3, pp. 1144--1158, Mar. 2016.

\bibitem{anderson2011}
T.~W. Anderson, \emph{The statistical analysis of time series}.\hskip 1em plus
  0.5em minus 0.4em\relax John Wiley \& Sons, 2011, vol.~19.

\bibitem{park2009spl}
H.~Park and M.~Van~der Schaar, ``On the impact of bounded rationality in
  peer-to-peer networks,'' \emph{IEEE Signal Processing Letters}, vol.~16,
  no.~8, pp. 675--678, 2009.

\bibitem{park2009tmm}
------, ``A framework for foresighted resource reciprocation in p2p networks,''
  \emph{IEEE Transactions on Multimedia}, vol.~11, no.~1, pp. 101--116, 2009.

\bibitem{bellman2015}
R.~E. Bellman and S.~E. Dreyfus, \emph{Applied dynamic programming}.\hskip 1em
  plus 0.5em minus 0.4em\relax Princeton university press, 2015.

\bibitem{pillai2005}
S.~U. Pillai, T.~Suel, and S.~Cha, ``The perron-frobenius theorem: some of its
  applications,'' \emph{IEEE Signal Processing Magazine}, vol.~22, no.~2, pp.
  62--75, 2005.

\bibitem{gross2004}
J.~L. Gross and J.~Yellen, \emph{Handbook of graph theory}.\hskip 1em plus
  0.5em minus 0.4em\relax CRC press, 2004.

\bibitem{IEEE80211ac}
``{Wireless LAN Medium Access Control (MAC) and Physical Layer (PHY)}
  specifications,'' \emph{IEEE Std. 802.11-2013}, 2013.

\bibitem{cisco80211ac}
``802.11ac: The fifth generation of wi-fi - technical white paper,''
  \emph{CISCO}, Mar. 2014.

\bibitem{miao2016}
G.~Miao, J.~Zander, K.~W. Sung, and S.~B. Slimane, \emph{Fundamentals of Mobile
  Data Networks}.\hskip 1em plus 0.5em minus 0.4em\relax Cambridge University
  Press, 2016.

\end{thebibliography}
%
\vspace{-2cm}
\begin{IEEEbiography}
    [{\includegraphics[width=1.in,height=1.1in,clip,keepaspectratio]{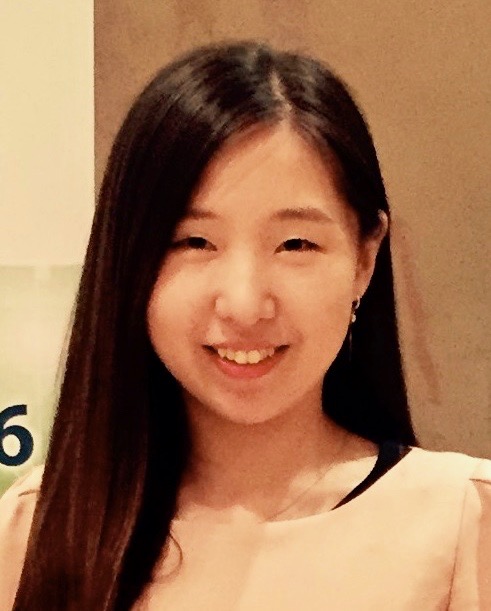}}]
    {Minhae Kwon}
received the B.S., M.S., and Ph.D. degrees at the Department of Electronic
and Electrical Engineering, Ewha Womans University, Seoul, Korea. 

Her research interest lies at the intersection of the network and distributed decision making; in particular, in using stochastic and data-driven approaches to capture dynamics and uncertainty of complex system. Relevant applications can be brain networks, wired/wireless networks, and autonomous networks/systems with multi-agents.

Dr. Kwon received the Minister's award of Science and ICT as well as Global Ph.D. Fellowship from Korea Government, and won honorable awards from Google, Qualcomm and IEEE Consumer Electronics Society.
\end{IEEEbiography}
\vspace{-2cm}

\begin{IEEEbiography}
    [{\includegraphics[width=1in,height=1.1in,clip,keepaspectratio]{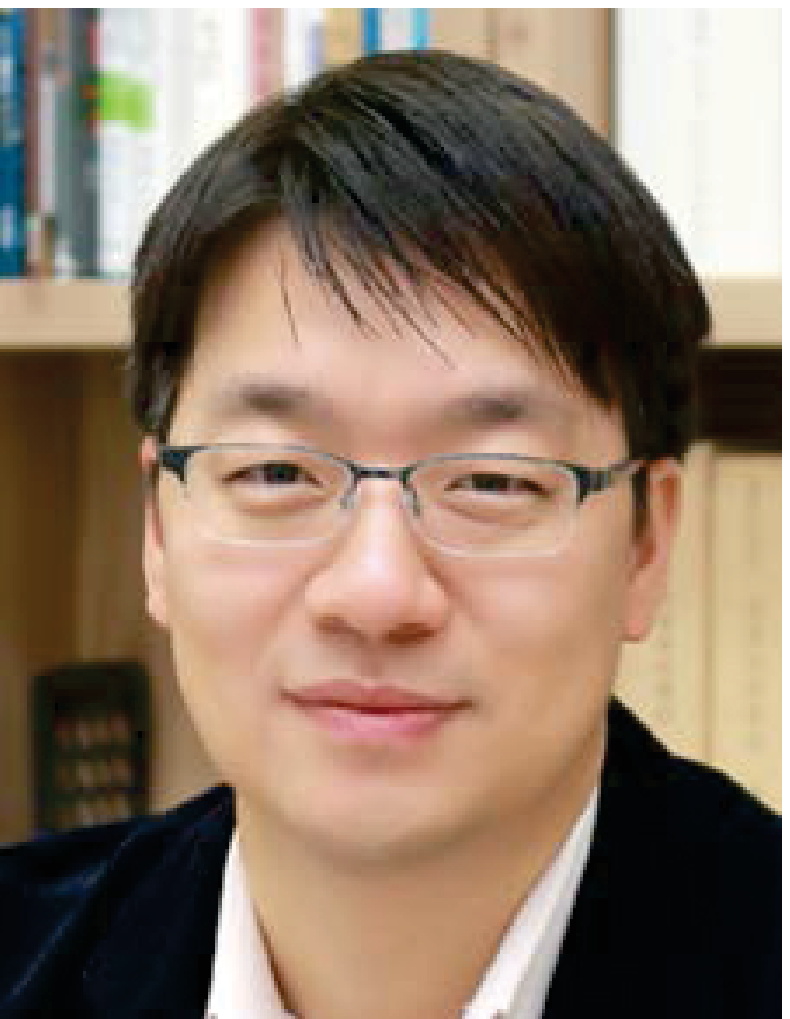}}]
    {Hyunggon Park}
received the B.S. degree in
Electronics and Electrical Engineering from
the Pohang University of Science and Technology
(POSTECH), Pohang, Korea, in 2004,
and the M.S. and Ph.D. degrees in Electrical
Engineering from the University of California,
Los Angeles (UCLA), in 2006 and 2008,
respectively. Currently, he is an Associate  
Professor at the Department of Electronic
and Electrical Engineering, Ewha Womans University, Seoul,
Korea. 

His research interests include 
machine learning based distributed decision making strategies for multi-agent network systems and 
efficient and robust data streaming strategies using network coding. In 2008, he was an intern at IBM T.J.Watson
Research Center, Hawthorne, NY, and he was a Senior Researcher at the
Signal Processing Laboratory (LTS4), Swiss Federal Institute of Technology
(EPFL), Lausanne, Switzerland, in 2009-2010. 

Dr. Park was a recipient of
the Graduate Study Abroad Scholarship from the Korea Science and
Engineering Foundation during 2004-2006 and a recipient of the Electrical
Engineering Department Fellowship at UCLA in 2008.
\end{IEEEbiography}

\newpage
\setcounter{page}{1}

 \onecolumn
 
\noindent{\Large [Supplemental Material]}

\begin{center}
{\huge Network Coding Based Evolutionary Network Formation for Dynamic Wireless
Networks}

\vspace{0.5cm}

{\Large Minhae Kwon and Hyunggon Park}

{\large \texttt{minhae.kwon@ewhain.net, hyunggon.park@ewha.ac.kr}}

\end{center}
 \vspace{0.5in}

\appendices


\section{Bellman Operation and its Properties}
Let $\mathcal T^*$ be the Bellman optimality operator~\cite{bellman2015} for $V_{\tau}(s)$, which maps a foresighted  state-value function to a foresighted  state-value function (i.e., $\mathcal T^*: \mathbb R^{|\mathbf S|} \rightarrow \mathbb R^{|\mathbf S|} $), defined as,
\begin{equation}
(\mathcal T^* V_{\tau})(s) = \max_{a \in \mathbf A} \left(  \sum_{s'\in \mathbf S} P(s'|s,a)  \left(  U(s,a,s') + \rho   V_{\tau}(s') \right) \right). 
\label{eqn:Bellman_op}
\end{equation}
This updates the state-value function with the action that provides the maximum expected long-term state value. 

The Bellman optimality operator $\mathcal T^*$ has monotonicity, additivity and $\rho$-contraction properties as shown below. 
\begin{property}
(Monotonicity of Bellman Optimality Operator)
$V_{\tau}(s) \le V_{\tau}'(s), \forall s  \Rightarrow (\mathcal T^* V_{\tau})(s) \le (\mathcal T^* V_{\tau}')(s)$
\label{prop:mono}
\end{property}

\begin{proof}
For the optimal action $a^* = \pi^* (s) = \arg\max_{a} \sum_{s\in \mathbf S} P(s'|s,a)  \left(  U(s,a,s') + \rho   V^*(s') \right)$ and when $V_{\tau}(s) \le V_{\tau}'(s), \forall s$, 
$(\mathcal T^* V_{\tau})(s) - (\mathcal T^* V_{\tau}')(s)$ becomes:  
\begin{align}
&(\mathcal T^* V_{\tau})(s) - (\mathcal T^* V_{\tau}')(s)\notag\\
&= \left( \sum_{s'\in \mathbf S} P(s'|s,a^*)  \left(  U(s,a^*,s') + \rho   V_{\tau}(s') \right)\right) \notag \\
&-  \left( \sum_{s'\in \mathbf S} P(s'|s,a^*)  \left(  U(s,a^*,s') + \rho   V_{\tau}'(s') \right)\right) \notag \\
&= \rho \sum_{s'\in \mathbf S} P(s'|s,a^*) \left(  V_{\tau}(s')  -  V_{\tau}'(s') \right)\label{eqn:monoto_1}\\
&\le 0 \label{eqn:monoto_2}
\end{align}
The inequality between \eqref{eqn:monoto_1} and \eqref{eqn:monoto_2} is satisfied because $0<\rho<1$, $ P(s'|s,a^*) \ge 0$, and $ V_{\tau}(s')  -  V_{\tau}'(s') \le 0$. 
Therefore, if $V_{\tau}(s) \le V_{\tau}'(s)$, then $(\mathcal T^* V_{\tau})(s) \le (\mathcal T^* V_{\tau}')(s)$.\end{proof}

\begin{property}
(Additivity of Bellman Optimality Operator)
$(\mathcal T^* V_{\tau} + d)(s) = (\mathcal T^* V_{\tau})(s) + \rho d, \forall s \in \mathbf S_{i}$
\label{prop:additivity}
\end{property}
\begin{proof}
\begin{align}
(\mathcal T^* &V_{\tau} + d)(s)\notag\\
 =&  \max_{a \in \mathbf A} \left( \sum_{s'\in \mathbf S} P(s'|s,a)  \left(  U(s,a,s') + \rho (  V_{\tau}(s')+d) \right)\right)\notag\\
=&  \max_{a \in \mathbf A} \left( \sum_{s'\in \mathbf S} P(s'|s,a)  \left(  U(s,a,s') + \rho  V_{\tau}(s') \right) \right.\\
&\left. 
+ \rho d \sum_{s'\in \mathbf S} P(s'|s,a)   \right) \label{eqn:additivity}
\\
=& (\mathcal T^* V_{\tau})(s) + \rho d \notag
\end{align}
where $\sum_{s'\in \mathbf S} P(s'|s,a)  =1$ in \eqref{eqn:additivity}. 
Therefore, $(\mathcal T^* V_{\tau} + d)(s) = (\mathcal T^* V_{\tau})(s) + \rho d, \forall s \in \mathbf S$.
\end{proof}

\begin{property}
($\rho$-Contraction Property of Bellman Optimality Operator)
$||\mathcal T^* V_{\tau}(s) -\mathcal T^* V_{\tau}'(s) ||_{\infty} \le \rho ||V_{\tau}(s) -V_{\tau}'(s)  ||_{\infty}, \forall s \in \mathbf S$   
\label{prop:contraction}
\end{property}
\begin{proof}
We define $d$  as, 
\begin{equation}
d = ||V_{\tau}(s) -V_{\tau}'(s) ||_{\infty},
\label{eqn:d}
\end{equation} 
where $|| \cdot ||_{\infty}$ denotes the infinite norm, defined as 
\begin{equation}
|| V_{\tau}(s) ||_{\infty} = \sup \left\{ |  V_{\tau}(s) |: s \in \mathbf S \right\} \label{eqn:infinite_norm}. 
\end{equation}
Then the following equations can be obtained from \eqref{eqn:d}.
\begin{align}
&V_{\tau}(s) - d \le V_{\tau}'(s) \le V_{\tau}(s) + d \label{eqn:contraction+d}
\\
&\mathcal T^*(V_{\tau}(s) - d) \le (\mathcal T^*V_{\tau}')(s) \le \mathcal T^* \left(V_{\tau}(s) + d \right) \label{eqn:before}\\
&\mathcal T^*V_{\tau}(s) - \rho d \le (\mathcal T^*V_{\tau}')(s) \le \mathcal T^* V_{\tau}(s) + \rho d \label{eqn:after}\\
& ||\mathcal T^*V_{\tau}(s) - \mathcal T^*V_{\tau}'(s)||_{\infty} \le \rho d \label{eqn:last}
\end{align}
The Bellman optimality operator is used between \eqref{eqn:contraction+d} and \eqref{eqn:before}, and Property~\ref{prop:additivity} is used between \eqref{eqn:before} and \eqref{eqn:after}. By substituting $d$ in \eqref{eqn:last} for \eqref{eqn:d}, we conclude the following equation. 
\begin{equation*}
||\mathcal T^* V_{\tau}(s) -\mathcal T^* V_{\tau}'(s) ||_{\infty} \le \rho ||V_{\tau}(s) -V_{\tau}'(s)  ||_{\infty}   
\end{equation*} 
\end{proof}


\section{Proof of Theorem~\ref{th:stop}}
\label{app:th:stop}
In this proof, we show that for all $s \in \mathbf S$, 
\begin{equation}
 || V_{\tau}(s) - V_{\tau-1}(s)||_{\infty} \le \frac{1-\rho}{2\rho} \epsilon \Rightarrow || V^{\epsilon^*}(s) - V^*(s) ||_{\infty} < \epsilon. 
 \label{eqn:pro}
\end{equation}
By using the definition of infinite norm, 
$ ||  V^{\epsilon^*}(s) - V^*(s) ||_{\infty}$ can be written as follow. 
\begin{align}
& || V^{\epsilon^*}(s) - V^*(s) ||_{\infty} \\
 &\le  ||  V^{\epsilon^*}(s) - V_{\tau}(s) ||_{\infty}+ ||  V_{\tau}(s) - V^*(s) ||_{\infty}
 \label{eqn:sum}
 \end{align}
 We now bound each part of the summation in \eqref{eqn:sum} individually:
 \begin{align}
  &||  V^{\epsilon^*}(s) - V_{\tau}(s) ||_{\infty} \\&=  || \mathcal T^{\epsilon^*} V^{\epsilon^*}(s) - V_{\tau}(s) ||_{\infty}\label{eqn:epsilon1}\\
  &\le || \mathcal T^{\epsilon^*} V^{\epsilon^*}(s) - \mathcal T^* V_{\tau}(s) ||_{\infty} + ||  \mathcal T^* V_{\tau}(s)  - V_{\tau}(s) ||_{\infty}\label{eqn:epsilon2}\\
  &=  || \mathcal T^{\epsilon^*} V^{\epsilon^*}(s) - \mathcal T^{\epsilon^*} V_{\tau}(s) ||_{\infty} + ||  \mathcal T^* V_{\tau}(s)  -  \mathcal T^* V_{\tau-1}(s) ||_{\infty}\label{eqn:epsilon3}\\
  & \le \rho ||  V^{\epsilon^*}(s) - V_{\tau}(s) ||_{\infty} + \rho ||  V_{\tau}(s)  - V_{\tau-1}(s) ||_{\infty} \label{eqn:epsilon4}\\
  &\le \frac{\rho}{1-\rho}||  V_{\tau}(s)  - V_{\tau-1}(s) ||_{\infty}  \notag
 \end{align}
 where $\mathcal T^{\epsilon^*}$ denotes the Bellman $\epsilon$-optimality operation and it satisfies $V^{\epsilon^*}(s)  = \mathcal T^{\epsilon^*}V^{\epsilon^*}(s) $ because $V^{\epsilon^*}(s) $ is the fixed point of $\mathcal T^{\epsilon^*}$, which is used in \eqref{eqn:epsilon1}. 
The inequality between \eqref{eqn:epsilon1} and \eqref{eqn:epsilon2} is obtained by using the definition of infinite norm. 
 Since $\pi^{\epsilon^*}$ is maximized  over the actions using $V_{\tau}(s)$ in~\eqref{eqn:epsilon3}, this  implies that $\mathcal T^{\epsilon^*}V_{\tau}(s) =\mathcal T^{*} V_{\tau}(s)$. 
The inequality between \eqref{eqn:epsilon3} and \eqref{eqn:epsilon4} is based on Property~\ref{prop:contraction}. 

Similarly, the second part of the summation in \eqref{eqn:sum} becomes
\begin{align}
&||  V_{\tau}(s) - V^*(s) ||_{\infty}\\
 &\le ||  V_{\tau}(s) - \mathcal T^{*} V_{\tau}(s) ||_{\infty}
 +||  \mathcal T^{*} V_{\tau}(s)- V^*(s) ||_{\infty}\notag\\
 &\le \rho ||  V_{\tau-1}(s) - V_{\tau}(s) ||_{\infty}
 +||  V_{\tau}(s)- V^*(s) ||_{\infty}\notag\\
 &\le \frac{\rho}{1-\rho}||  V_{\tau}(s) - V_{\tau-1}(s) ||_{\infty}\label{eqn:epsilon5}. 
 \end{align}
By substituting \eqref{eqn:epsilon4} and \eqref{eqn:epsilon5} in \eqref{eqn:sum} and using the condition $ || V_{\tau} - V_{\tau-1}||_{\infty} \le \frac{1-\rho}{2\rho} \epsilon$ in \eqref{eqn:pro}, we conclude the following. 
\begin{align*}
 || V^{\epsilon^*}(s) - V^*(s) ||_{\infty}  &\le \frac{2\rho}{1-\rho}||  V_{\tau}(s)  - V_{\tau-1}(s) ||_{\infty} \\
 &<  \frac{2\rho}{1-\rho}  \frac{1-\rho}{2\rho} \epsilon\\
&= \epsilon
\end{align*}
Therefore, for all $s \in \mathbf S$, \eqref{eqn:pro} is satisfied. 


\section{Proof of Theorem~\ref{th:value_convergence}}
\label{app:th:value_convergence}
In this proof, we show that Algorithm~\ref{alg:value_it} converges to the optimal policies $\pi^*$ by showing that the infinite interactions of the Bellman optimality operation converge to the optimal state-value function $V^*(s)$, such as 
\begin{equation}
\lim_{\tau \rightarrow \infty}  V_{\tau}(s)  = V^*(s), \forall s \in \mathbf S.  \label{eqn:th1}
\end{equation}
This is identical  to the following equation based on the definition of the infinite norm in~\eqref{eqn:infinite_norm}. 
\begin{equation}
\lim_{\tau \rightarrow \infty} || V_{\tau}(s) - V^*(s)||_{\infty} = 0, \forall s \in \mathbf S  \label{eqn:alt}
\end{equation} 
Hence, in this proof, we prove \eqref{eqn:alt} as below. 
\begin{align}
&\lim_{\tau \rightarrow \infty}||V_{\tau}(s) - V^*(s)||_{\infty}\\
 &= \lim_{\tau \rightarrow \infty} ||\mathcal T^* V_{{\tau}-1}(s) -\mathcal T^* V^*(s)||_{\infty} \label{eqn:th_proof1}\\
&\le \lim_{\tau \rightarrow \infty} \rho   ||V_{\tau-2}(s) - V^*(s)||_{\infty}\label{eqn:th_proof2}\\ 
&\le  \cdots \le \lim_{\tau \rightarrow \infty} \rho^{\tau} ||V_0(s) - V^*(s)||_{\infty}\label{eqn:th_proof3}\\
&= 0\notag
\end{align}
Based on the definition of $V^*(s)$ in \eqref{eqn:opt_value}, $V^*(s) = \mathcal T^* V^*(s)$ is used in \eqref{eqn:th_proof1}, and  the inequality between \eqref{eqn:th_proof1} and \eqref{eqn:th_proof2} is based on Property~\ref{prop:contraction}. 
Since $0<\rho<1$, $\lim_{\rho \rightarrow \infty} \rho^{\tau} = 0$ in \eqref{eqn:th_proof3}. 
Therefore, $\lim_{\tau \rightarrow \infty} || V_{\tau}(s) - V^*(s)||_{\infty} = 0, \forall s \in \mathbf S$ so that  $\lim_{\tau \rightarrow \infty}  V_{\tau}(s)  = V^*(s), \forall s_{i,\tau} \in \mathbf S$.

\end{document}